\documentclass[sigconf,nonacm]{acmart} %
\setcopyright{none}

\usepackage{import}
\usepackage{tikz}
\usepackage{amsmath}

\usepackage[utf8]{inputenc}

\input{_boilerplate_base.tex}

\usepackage[dvipsnames]{xcolor}

\definecolor{myA16zGrayLight}{RGB}{235,235,235}     %
\definecolor{myA16zGrayMedium}{RGB}{196,196,196}    %
\definecolor{myA16zGrayDark}{RGB}{44,34,34}         %
\definecolor{myA16zLavender}{RGB}{208,161,255}      %
\definecolor{myA16zMagenta}{RGB}{195,70,206}        %
\definecolor{myA16zMulberry}{RGB}{113,24,88}        %
\definecolor{myA16zLemonChiffon}{RGB}{250,234,157}  %
\definecolor{myA16zAmber}{RGB}{230,154,48}          %
\definecolor{myA16zRust}{RGB}{174,59,10}            %
\definecolor{myA16zLime}{RGB}{197,222,107}          %
\definecolor{myA16zAquamarine}{RGB}{82,216,145}     %
\definecolor{myA16zPine}{RGB}{60,87,44}             %
\definecolor{myA16zPacific}{RGB}{145,224,235}       %
\definecolor{myA16zTeal}{RGB}{36,197,201}           %
\definecolor{myA16zAzure}{RGB}{18,51,90}            %

\definecolor{myTechnionDeepBlue}{HTML}{002147}           %
\definecolor{myTechnionGoldenOchre}{HTML}{D59F0F}        %
\definecolor{myTechnionBlack}{HTML}{000000}              %
\definecolor{myTechnionWhite}{HTML}{FFFFFF}              %

\definecolor{myTechnionRed}{HTML}{E31D1A}             %
\definecolor{myTechnionPink}{HTML}{EA094B}           %
\definecolor{myTechnionPurple1}{HTML}{AE3B72}         %
\definecolor{myTechnionPurple2}{HTML}{4D4084}         %
\definecolor{myTechnionBlue1}{HTML}{216093}            %
\definecolor{myTechnionBlue2}{HTML}{5686DA}           %
\definecolor{myTechnionTeal}{HTML}{32B1CA}            %
\definecolor{myTechnionGreen1}{HTML}{EA094B}           %
\definecolor{myTechnionGreen2}{HTML}{A3D65C}           %
\definecolor{myTechnionGreen3}{HTML}{94D60A}           %
\definecolor{myTechnionYellow}{HTML}{FDD700}          %
\definecolor{myTechnionOrange}{HTML}{FF6B00}         %
\definecolor{myTechnionBrown}{HTML}{97775C}          %
\definecolor{myTechnionBeige}{HTML}{D9D1C3}          %
\definecolor{myTechnionGray1}{HTML}{A2A9AE}            %
\definecolor{myTechnionGray2}{HTML}{5A6771}            %

\definecolor{mySuCardinalRed}{HTML}{8c1515}
\definecolor{mySuCardinalRedLight}{HTML}{B83A4B}
\definecolor{mySuCardinalRedDark}{HTML}{820000}
\definecolor{mySuWhite}{HTML}{ffffff}
\definecolor{mySuCoolGrey}{HTML}{53565A}
\definecolor{mySuBlack}{HTML}{2e2d29}
\definecolor{mySuBlack100}{HTML}{2e2d29}
\definecolor{mySuBlack90}{HTML}{43423E}
\definecolor{mySuBlack80}{HTML}{585754}
\definecolor{mySuBlack70}{HTML}{6D6C69}
\definecolor{mySuBlack60}{HTML}{767674}
\definecolor{mySuBlack50}{HTML}{979694}
\definecolor{mySuBlack40}{HTML}{ABABA9}
\definecolor{mySuBlack30}{HTML}{C0C0BF}
\definecolor{mySuBlack20}{HTML}{D5D5D4}
\definecolor{mySuBlack10}{HTML}{EAEAEA}

\definecolor{mySuPaloAlto}{HTML}{175E54}
\definecolor{mySuPaloAltoLight}{HTML}{2D716F}
\definecolor{mySuPaloAltoDark}{HTML}{014240}
\definecolor{mySuPaloVerde}{HTML}{279989}
\definecolor{mySuPaloVerdeLight}{HTML}{59B3A9}
\definecolor{mySuPaloVerdeDark}{HTML}{017E7C}
\definecolor{mySuOlive}{HTML}{8F993E}
\definecolor{mySuOliveLight}{HTML}{A6B168}
\definecolor{mySuOliveDark}{HTML}{7A863B}
\definecolor{mySuBay}{HTML}{6FA287}
\definecolor{mySuBayLight}{HTML}{8AB8A7}
\definecolor{mySuBayDark}{HTML}{417865}
\definecolor{mySuSky}{HTML}{4298B5}
\definecolor{mySuSkyLight}{HTML}{67AFD2}
\definecolor{mySuSkyDark}{HTML}{016895}
\definecolor{mySuLagunita}{HTML}{007C92}
\definecolor{mySuLagunitaLight}{HTML}{009AB4}
\definecolor{mySuLagunitaDark}{HTML}{006B81}
\definecolor{mySuPoppy}{HTML}{E98300}
\definecolor{mySuPoppyLight}{HTML}{F9A44A}
\definecolor{mySuPoppyDark}{HTML}{D1660F}
\definecolor{mySuSpirited}{HTML}{E04F39}
\definecolor{mySuSpiritedLight}{HTML}{F4795B}
\definecolor{mySuSpiritedDark}{HTML}{C74632}
\definecolor{mySuIlluminating}{HTML}{FEDD5C}
\definecolor{mySuIlluminatingLight}{HTML}{FFE781}
\definecolor{mySuIlluminatingDark}{HTML}{FEC51D}
\definecolor{mySuPlum}{HTML}{620059}
\definecolor{mySuPlumLight}{HTML}{734675}
\definecolor{mySuPlumDark}{HTML}{350D36}
\definecolor{mySuBrick}{HTML}{651C32}
\definecolor{mySuBrickLight}{HTML}{7F2D48}
\definecolor{mySuBrickDark}{HTML}{42081B}
\definecolor{mySuArchway}{HTML}{5D4B3C}
\definecolor{mySuArchwayLight}{HTML}{766253}
\definecolor{mySuArchwayDark}{HTML}{2F2424}
\definecolor{mySuStone}{HTML}{7F7776}
\definecolor{mySuStoneLight}{HTML}{D4D1D1}
\definecolor{mySuStoneDark}{HTML}{544948}
\definecolor{mySuFog}{HTML}{DAD7CB}
\definecolor{mySuFogLight}{HTML}{F4F4F4}
\definecolor{mySuFogDark}{HTML}{B6B1A9}

\definecolor{mySuDigitalRed}{HTML}{B1040E}
\definecolor{mySuDigitalRedLight}{HTML}{E50808}
\definecolor{mySuDigitalRedDark}{HTML}{820000}
\definecolor{mySuDigitalBlue}{HTML}{006CB8}
\definecolor{mySuDigitalBlueLight}{HTML}{6FC3FF}
\definecolor{mySuDigitalBlueDark}{HTML}{00548f}
\definecolor{mySuDigitalGreen}{HTML}{008566}
\definecolor{mySuDigitalGreenLight}{HTML}{1AECBA}
\definecolor{mySuDigitalGreenDark}{HTML}{006F54}

\definecolor{myParula1Blue}{RGB}{0,114,189}
\definecolor{myParula2Orange}{RGB}{217,83,25}
\definecolor{myParula3Yellow}{RGB}{237,177,32}
\definecolor{myParula4Purple}{RGB}{126,47,142}
\definecolor{myParula5Green}{RGB}{119,172,48}
\definecolor{myParula6LightBlue}{RGB}{77,190,238}
\definecolor{myParula7Red}{RGB}{162,20,47}

\usepackage{tikz}
\usetikzlibrary{calc}
\usetikzlibrary{arrows.meta}
\usetikzlibrary{patterns}
\usetikzlibrary{positioning}
\usetikzlibrary{decorations.pathreplacing}
\usetikzlibrary{decorations.pathmorphing}
\usetikzlibrary{calligraphy}
\usetikzlibrary{shapes.misc}
\usetikzlibrary{shapes.symbols}
\usetikzlibrary{shapes.arrows}
\usetikzlibrary{spy}
\usetikzlibrary{shapes.geometric}
\usetikzlibrary{intersections}
\usetikzlibrary{fit}

\usepackage{pgfplots}
\pgfplotsset{compat=1.17}
\usepgfplotslibrary{fillbetween}

\pgfplotsset{
    discard if not/.style 2 args={
            x filter/.code={
                    \edef\tempa{\thisrow{#1}}
                    \edef\tempb{#2}
                    \ifx\tempa\tempb
                    \else
                        \def\pgfmathresult{inf}
                    \fi
                }
        },
}

\pgfplotsset{
    mysimpleplot/.style = {
            every axis plot/.prefix style={thick},
            width=1.05\linewidth,
            height=0.75\linewidth,
            title style={font=\scriptsize,align=center},
            legend cell align=left,
            legend style={font=\scriptsize},
            legend columns=3,
            legend style={
                    at={(0.5,1)},
                    yshift=0.3em,
                    anchor=south,
                    draw=none,
                    /tikz/every even column/.append style={
                            column sep=0.3em
                        },
                    cells={
                            align=left
                        }
                },
            grid=both,
            minor tick num=9,
            major grid style={solid,very thin,draw=gray!50},
            minor grid style={solid,ultra thin,draw=gray!20},
            label style={font=\scriptsize,align=center},
            tick label style={font=\scriptsize},
        },
}

\pgfplotsset{
    mysimplefig1plot/.style = {
        mysimpleplot,
        xlabel={$\netX$},
        ylabel={$\tALident/n$},
        xmin=0.0, xmax=0.5,
        ymin=0.0, ymax=0.35,
        height=0.75\linewidth,
        width=\linewidth,
        yticklabel style={
                /pgf/number format/fixed,
                /pgf/number format/precision=2
            },
        scaled y ticks=false,
        xtick={0,0.1,0.2,0.25,0.33333,0.5},
        xticklabels={0,1/10,1/5,1/4,1/3,1/2},
        yticklabels={0,1/10,1/6,1/5,1/4,1/3},
        ytick={0,0.1,0.16666,0.2,0.25,0.33333},
    }
}

\tikzset{myparula11/.style={color=myParula1Blue,solid,mark=+,mark options={solid}}}
\tikzset{myparula12/.style={color=myParula1Blue,densely dashed,mark=x,mark options={solid}}}
\tikzset{myparula13/.style={color=myParula1Blue,densely dotted,mark=o,mark options={solid}}}
\tikzset{myparula14/.style={color=myParula1Blue,dashdotted,mark=triangle,mark options={solid}}}
\tikzset{myparula15/.style={color=myParula1Blue,dashdotdotted,mark=square,mark options={solid}}}

\tikzset{myparula21/.style={color=myParula2Orange,solid,mark=+,mark options={solid}}}
\tikzset{myparula22/.style={color=myParula2Orange,densely dashed,mark=x,mark options={solid}}}
\tikzset{myparula23/.style={color=myParula2Orange,densely dotted,mark=o,mark options={solid}}}
\tikzset{myparula24/.style={color=myParula2Orange,dashdotted,mark=triangle,mark options={solid}}}
\tikzset{myparula25/.style={color=myParula2Orange,dashdotdotted,mark=square,mark options={solid}}}

\tikzset{myparula31/.style={color=myParula3Yellow,solid,mark=+,mark options={solid}}}
\tikzset{myparula32/.style={color=myParula3Yellow,densely dashed,mark=x,mark options={solid}}}
\tikzset{myparula33/.style={color=myParula3Yellow,densely dotted,mark=o,mark options={solid}}}
\tikzset{myparula34/.style={color=myParula3Yellow,dashdotted,mark=triangle,mark options={solid}}}
\tikzset{myparula35/.style={color=myParula3Yellow,dashdotdotted,mark=square,mark options={solid}}}

\tikzset{myparula41/.style={color=myParula4Purple,solid,mark=+,mark options={solid}}}
\tikzset{myparula42/.style={color=myParula4Purple,densely dashed,mark=x,mark options={solid}}}
\tikzset{myparula43/.style={color=myParula4Purple,densely dotted,mark=o,mark options={solid}}}
\tikzset{myparula44/.style={color=myParula4Purple,dashdotted,mark=triangle,mark options={solid}}}
\tikzset{myparula45/.style={color=myParula4Purple,dashdotdotted,mark=square,mark options={solid}}}

\tikzset{myparula51/.style={color=myParula5Green,solid,mark=+,mark options={solid}}}
\tikzset{myparula52/.style={color=myParula5Green,densely dashed,mark=x,mark options={solid}}}
\tikzset{myparula53/.style={color=myParula5Green,densely dotted,mark=o,mark options={solid}}}
\tikzset{myparula54/.style={color=myParula5Green,dashdotted,mark=triangle,mark options={solid}}}
\tikzset{myparula55/.style={color=myParula5Green,dashdotdotted,mark=square,mark options={solid}}}

\tikzset{myparula61/.style={color=myParula6LightBlue,solid,mark=+,mark options={solid}}}
\tikzset{myparula62/.style={color=myParula6LightBlue,densely dashed,mark=x,mark options={solid}}}
\tikzset{myparula63/.style={color=myParula6LightBlue,densely dotted,mark=o,mark options={solid}}}
\tikzset{myparula64/.style={color=myParula6LightBlue,dashdotted,mark=triangle,mark options={solid}}}
\tikzset{myparula65/.style={color=myParula6LightBlue,dashdotdotted,mark=square,mark options={solid}}}

\tikzset{myparula71/.style={color=myParula7Red,solid,mark=+,mark options={solid}}}
\tikzset{myparula72/.style={color=myParula7Red,densely dashed,mark=x,mark options={solid}}}
\tikzset{myparula73/.style={color=myParula7Red,densely dotted,mark=o,mark options={solid}}}
\tikzset{myparula74/.style={color=myParula7Red,dashdotted,mark=triangle,mark options={solid}}}
\tikzset{myparula75/.style={color=myParula7Red,dashdotdotted,mark=square,mark options={solid}}}

\usepackage{multirow}
\usepackage{booktabs}   %
\usepackage{makecell}
\usepackage{threeparttable}   %

\usepackage{algorithm}
\usepackage{algorithmicx}
\usepackage[noend]{algpseudocode}

\makeatletter
\AddToHook{env/algorithmic/begin}{\def\@currentcounter{ALG@line}}
\makeatother

\algnewcommand{\LineComment}[1]{\State {\textcolor{gray}{/\!/ #1}}}

\newcommand{\alglocref}[2]{\cref{#1}, \cref{#2}}

\algrenewcommand{\alglinenumber}[1]{\scriptsize\textcolor{gray}{\texttt{#1}}}
\algrenewcommand{\algorithmicindent}{1em}

\algnewcommand{\algfontsize}[0]{\footnotesize}
\AddToHook{env/algorithmic/begin}{\algfontsize}

\algnewcommand{\algorithmicswitch}{\textbf{switch}}
\algdef{SE}[SWITCH]{Switch}{EndSwitch}[1]{\algorithmicswitch\ #1\ \algorithmicdo}{\algorithmicend\ \algorithmicswitch}%
\algtext*{EndSwitch}%

\algnewcommand{\algorithmiccase}{\textbf{case}}
\algdef{SE}[CASE]{Case}{EndCase}[1]{\algorithmiccase\ #1}{\algorithmicend\ \algorithmiccase}%
\algtext*{EndCase}%

\algnewcommand{\algorithmicon}{\textbf{on}}
\algdef{SE}[ON]{On}{EndOn}[1]{\algorithmicon\ #1\ \algorithmicdo}{\algorithmicend\ \algorithmicon}%
\algtext*{EndOn}%

\algnewcommand{\algorithmicupon}{\textbf{upon}}
\algdef{SE}[UPON]{Upon}{EndUpon}[1]{\algorithmicupon\ #1\ \algorithmicdo}{\algorithmicend\ \algorithmicupon}%
\algtext*{EndUpon}%

\algnewcommand{\algorithmicat}{\textbf{at}}
\algdef{SE}[AT]{At}{EndAt}[1]{\algorithmicat\ #1\ \algorithmicdo}{\algorithmicend\ \algorithmicat}%
\algtext*{EndAt}%

\algnewcommand{\algorithmicrealfunction}{\textbf{function}}
\algdef{SE}[REALFUNCTION]{RealFunction}{EndRealFunction}[1]{\algorithmicrealfunction\ #1\ \algorithmicdo}{\algorithmicend\ \algorithmicrealfunction}%
\algtext*{EndRealFunction}%

\algnewcommand{\algorithmicthroughout}{\textbf{throughout}}
\algdef{SE}[Throughout]{Throughout}{EndThroughout}[1]{\algorithmicthroughout\ #1\ \algorithmicdo}{\algorithmicend\ \algorithmicthroughout}%
\algtext*{EndThroughout}%

\algnewcommand{\algorithmicforever}{\textbf{forever}}
\algdef{SE}[Forever]{Forever}{EndForever}[0]{\algorithmicforever\algorithmicdo}{\algorithmicend\ \algorithmicforever}%
\algtext*{EndForever}%

\algrenewcommand{\algorithmicdo}{}
\algrenewcommand{\algorithmicthen}{}

\algnewcommand{\algorithmicgoto}{\textbf{goto}}%
\algnewcommand{\Goto}[1]{\algorithmicgoto~\ref{#1}}%

\algnewcommand{\algorithmicassert}{\textbf{assert}}%
\algnewcommand{\Assert}[1]{\algorithmicassert~{#1}}%

\algnewcommand{\algorithmicbreak}{\textbf{break}}%
\algnewcommand{\Break}[0]{\algorithmicbreak}%

\algnewcommand{\algorithmicwaiton}{\textbf{wait on}}%
\algnewcommand{\WaitOn}[1]{\algorithmicwaiton~{#1}}%

\algnewcommand{\InlineRequire}[1]{\textbf{require} {#1}}

\algblock{ManualIndent}{EndManualIndent}
\algnotext{ManualIndent}
\algnotext{EndManualIndent}

\algdef{SE}[GENERICBLOCK]{GenericBlock}{EndGenericBlock}[1]{#1}{}%
\algtext*{EndGenericBlock}%

\usepackage{fontawesome5}
\usepackage{noto-emoji-easy}
\usepackage{tango-icons}

\subimport{./lib/}{defer.tex}

\makeatletter
\@ifpackageloaded{cleveref}{}{%
  \usepackage[sort&compress,capitalize,nameinlink]{cleveref}%
}
\makeatother

\crefalias{ALG@line}{line}

\crefname{figure}{Fig.}{Figs.}
\Crefname{figure}{Fig.}{Figs.}

\crefname{table}{Tab.}{Tabs.}
\Crefname{table}{Tab.}{Tabs.}

\crefname{section}{Sec.}{Secs.}
\Crefname{section}{Sec.}{Secs.}
\crefname{subsection}{Sec.}{Secs.}
\Crefname{subsection}{Sec.}{Secs.}
\crefname{subsubsection}{Sec.}{Secs.}
\Crefname{subsubsection}{Sec.}{Secs.}
\crefname{subsubsubsection}{Sec.}{Secs.}
\Crefname{subsubsubsection}{Sec.}{Secs.}
\crefname{appendix}{Appendix}{Appendices}
\Crefname{appendix}{Appendix}{Appendices}
\crefname{subappendix}{Appendix}{Appendices}
\Crefname{subappendix}{Appendix}{Appendices}
\crefname{subsubappendix}{Appendix}{Appendices}
\Crefname{subsubappendix}{Appendix}{Appendices}
\crefname{subsubsubappendix}{Appendix}{Appendices}
\Crefname{subsubsubappendix}{Appendix}{Appendices}

\crefname{algorithm}{Alg.}{Algs.}
\Crefname{algorithm}{Alg.}{Algs.}
\crefname{line}{ln.}{lns.}
\Crefname{line}{ln.}{lns.}

\crefname{proposition}{Prop.}{Props.}
\Crefname{proposition}{Prop.}{Props.}
\crefname{lemma}{Lem.}{Lems.}
\Crefname{lemma}{Lem.}{Lems.}
\crefname{theorem}{Thm.}{Thms.}
\Crefname{theorem}{Thm.}{Thms.}
\crefname{corollary}{Cor.}{Cors.}
\Crefname{corollary}{Cor.}{Cors.}
\crefname{definition}{Def.}{Defs.}
\Crefname{definition}{Def.}{Defs.}
\crefname{conjecture}{Conj.}{Conjs.}
\Crefname{conjecture}{Conj.}{Conjs.}
\crefname{remark}{Rem.}{Rems.}
\Crefname{remark}{Rem.}{Rems.}

\import{./}{commands.tex}
\newdefergroup{proof}[deferred]

\title{Gatling: Rapid-Fire Consensus from Parallel Composition} %

\author{Giulia Scaffino}
\affiliation{%
  \institution{TU Wien, Common Prefix}
  \city{}
  \country{}}
\email{giulia.scaffino@gmail.com}

\author{Max Resnick}
\affiliation{%
  \institution{Anza}
  \city{}
  \country{}}
\email{max.resnick@anza.xyz}

\author{Joachim Neu}
\affiliation{%
  \institution{a16z Crypto Research}
  \city{}
  \country{}}
\email{jneu@a16z.com}

\begin{document}
\begin{abstract}
Consensus protocols form the core of blockchains and other replicated state machines, ensuring that all correct nodes process the same totally ordered log of input transactions. In fault-free executions, performance is driven by the good-case transaction latency---the time between a transaction becoming known to all nodes and its confirmation by the consensus protocol---which depends on both how frequently proposals are made and, once made, how quickly they are confirmed. While prior work has established tight lower bounds on confirmation latency that modern protocols already achieve, it remains open whether the inter-proposal time can be further reduced below the state-of-the-art of one network delay.

We introduce \composite, an atomic broadcast protocol that achieves arbitrarily small inter-proposal times under rotating leader schedules; in particular, smaller than the network delay. \composite runs multiple parallel instances of a black-box atomic broadcast protocol and staggers their proposal schedules to generate proposals in faster succession than state-of-the-art protocols.
A deterministic interleaving rule merges the outputs of these instances into a single global log. We analyze the effects of head-of-line blocking caused by crashed leaders, and derive \composite's optimal number of parallel instances. We further study the impact of \composite on predictable validity and present two variants that retain this property.
Finally, our experiments confirm that \composite can be used with off-the-shelf component protocols to achieve low latency without fine-tuning the component protocol for minimum latency.
\end{abstract}

\maketitle

\section{Introduction}
\label{sec:introduction}

After the launch of the Bitcoin network in 2009, block times fell exponentially
for fifteen years. Bitcoin 2009: 10 minutes. Ethereum 2015: 15 seconds. Solana
2020: 400\,ms~\cite{nakamoto2008bitcoin,buterin2014ethereum,yakovenko2018solana}.
However, more than five years later,
blockchains
still launch
with almost the same block times Solana
had in 2020. Why did fifteen years of exponential progress, which promised to
drive block times to zero, suddenly stall in 2020?

The reason is structural. All of the blockchains listed above proceed one block
at a time. Each block references its parent, and so before consensus can begin
on the next block, the previous block must be known. Roughly speaking, this bounds block times to
the time it takes for a single round of consensus---and that time
has tight, well-known lower bounds.
Prior work~\cite{good-case-latency} proved a $3\Delta$ confirmation bound under
partial synchrony in the $3f+1 \leq n < 5f-1$ regime, and a smaller bound
holds at $n \geq 5f+1$, achieved by Minimmit~\cite{minimmit} at the cost of
resilience. The remaining knob is the inter-proposal time. Pipelining (e.g., in
Moonshot~\cite{moonshot} and Hydrangea++~\cite{hydrangeaplusplus}) pushed it to
$\Delta$. Hydrangea++ asserts, without proof, that $\Delta$ is the floor.
Intuitively, they argue that the next leader cannot propose a block without
seeing the previous block, and that takes at least $\Delta$ time. This argument
is of course correct; however, in this paper we challenge the assumption that
the next leader must see the previous block before proposing their own.

We present \composite, an atomic broadcast protocol whose inter-proposal time
can be made arbitrarily small, in particular smaller than the network delay
$\Delta$. \composite runs $K$ parallel instances of a closed-box atomic
broadcast protocol with inter-proposal time $\deltaibt$, staggered by $\deltaibt/K$, and interleaves their outputs
into a single log via a deterministic merge rule. Each component instance
still respects the $3\Delta$ confirmation bound; the speedup comes from
composition. Setting $K$ large drives the inter-proposal time toward zero. As
a proof of concept, we instantiate \composite on top of a Simplex~\cite{simplex}
implementation and obtain $244$\,ms end-to-end transaction latency and
$50$\,ms slot times on a global cluster with nodes in Asia, Europe, and North
America.

\myparagraph{Setting}
Blockchains are distributed systems whose operation relies on atomic broadcast
protocols to ensure consistent replication of state across participating nodes.
In such protocols, nodes receive \emph{transactions} as inputs from the
environment---an abstraction representing all entities external to the
protocol---and, after some time, they output to the environment a totally
ordered sequence of \emph{confirmed} transactions, commonly referred to as the
\emph{log}. At their core, atomic broadcast protocols typically consist of a
proposal phase, in which a designated node proposes a payload of transactions
to be appended to the log, followed by a number of voting phases (often two),
during which nodes collectively decide whether to confirm the proposal. Phases
repeat over time, allowing the protocol to extend the log incrementally as new
transactions arrive.

In atomic broadcast executions in which all nodes follow the protocol
(\emph{good-case}), a key performance metric is the \emph{average transaction
latency}, defined as the average time required for a transaction to be
confirmed after it has been received by all nodes. The good-case average
transaction latency, hereafter for brevity referred to as \emph{transaction
latency}, naturally decomposes into two components: the \emph{inter-proposal
time}, which captures the time until a transaction is included in a proposal,
and the \emph{confirmation time}, which measures how long a proposal takes to
become confirmed. Minimizing transaction latency therefore requires reducing
both components simultaneously. \Cref{tab:protocol-comparison} summarizes the
inter-proposal time, confirmation time, and transaction latency of several
state-of-the-art protocols. The remainder of this section unpacks where each
component's lower bound comes from and how prior protocols have approached it.

\begin{table*}[t]
    \centering
    \small
    \setlength{\tabcolsep}{6pt}
    \caption{
    Inter-proposal time ($\deltaibt$), confirmation time ($\deltaconf$), and average good-case transaction latency of different protocols. \composite is parameterized by $K$ such that $\epsilon  = 2\Delta/K$
    can be made arbitrarily small. Unlike the other protocols, \composite is a closed-box construction: it uses the underlying protocol only through its interface, without regard to internals.
    }
    \label{tab:protocol-comparison}
    \begin{threeparttable}
        \begin{tabular*}{\textwidth}{@{\extracolsep{\fill}}lcccc}
            \toprule
            \textbf{Protocol} & \textbf{\makecell{Inter-Proposal \\ Time}} & \textbf{\makecell{Confirmation \\ Time}} & \textbf{\makecell{Transaction \\ Latency}} & \textbf{\makecell{Closed \\ Box}} \\ %
            \midrule
            PBFT~\cite{pbft}, Tendermint~\cite{buchman2016tendermint}  & 3$\Delta$ & 3$\Delta$ & $4.5\Delta$  & No \textcolor{red}{\faTimes} \\
            HotStuff~\cite{hotstuff} & 2$\Delta$ & 7$\Delta$ \tnote{a} & $8\Delta$ &  No \textcolor{red}{\faTimes}  \\
            HotStuff-1~\cite{hotstuff-1}, HotStuff-2~\cite{hotstuff-2}, Jolteon~\cite{jolteon} & 2$\Delta$ & 5$\Delta$ \tnote{b} & $6\Delta$ &  No \textcolor{red}{\faTimes} \\
            Simplex~\cite{simplex} & 2$\Delta$ & 3$\Delta$ & $4\Delta$ & No \textcolor{red}{\faTimes} \\
            Moonshot~\cite{moonshot}, Hydrangea++~\cite{hydrangeaplusplus} & $\Delta$ & 3$\Delta$ \tnote{c} & $3.5\Delta$ &  No \textcolor{red}{\faTimes}   \\
            Shoal++~\cite{shoalplusplus}\tnote{d} & $\Delta$ & 4.5$\Delta$ & $5\Delta$ & No \textcolor{red}{\faTimes} \\
            \composite (this work)& $\epsilon$ & 3$\Delta$ & $(3 + \epsilon/2)\Delta$ & Yes \textcolor{green!70!black}{\faCheck} \\
            \bottomrule
        \end{tabular*}
        \begin{tablenotes}[para,flushleft]
            \footnotesize\noindent
            \tnote{a}\,HotStuff achieves a minimum commit latency of $7\Delta$ when the next leader aggregates the votes for the current leader's proposal.
            \tnote{b}\,HotStuff-1 has a $5\Delta$ confirmation time on the regular path, and an optimistic (speculative) path that reduces latency to $3\Delta$ by allowing nodes to send early execution (confirmation) responses to clients after a single QC. Clients need to handle responses with care, to avoid accepting provisional confirmations that may later be revoked or reordered.
            \tnote{c}\,We ignore the $20\%$-resilience fast path of Hydrangea++ as this is an orthogonal technique.
            \tnote{d}\,After completion of this work, it was brought to our attention that Shoal++~\cite{shoalplusplus} employs a construction similar to \composite in the context of DAG-based protocols to reduce the inter-proposal time.
        \end{tablenotes}
    \end{threeparttable}
\end{table*}

\myparagraph{Confirmation Time}
Prior work~\cite{good-case-latency} fully characterizes the tight good-case
bounds on confirmation time in the authenticated setting across synchronous,
partially synchronous, and asynchronous network models. Under partial synchrony
with $3f+1 \leq n < 5f-1$, the tight lower bound is three communication
rounds. Traditional protocols such as PBFT~\cite{pbft},
Tendermint~\cite{buchman2016tendermint}, and Simplex~\cite{simplex} operate in
conservative rounds of duration $\Delta$ and saturate this bound at $3\Delta$.
The bound therefore offers no further room for improvement at the standard
resilience threshold: any reduction in transaction latency must come from the
inter-proposal time.

\myparagraph{Inter-Proposal Time: Pipelining}
Traditional atomic broadcast protocols operate in discrete slots (or views),
each associated with a leader pseudo-randomly elected from the participating
nodes (rotating leaders). The leader of the slot is entitled to cast a proposal
for its slot. Typically, proposals are strictly sequential: each proposal must
complete all voting phases before the next can be issued. In PBFT and
Tendermint, for instance, a leader proposes a block and waits for two full
rounds of voting, yielding a ``propose-vote-vote'' structure and an
inter-proposal time of $3\Delta$. The HotStuff family of
protocols~\cite{hotstuff,hotstuff-1,hotstuff-2} and Jolteon~\cite{jolteon}
relax this strict sequentiality by cryptographically chaining proposals with
quorum certificates (QCs), such that each proposal includes a QC that carries
votes for the previous proposal: as a result, a leader can issue a new proposal
while nodes are still voting on the previous one. This technique, called
pipelining, produces an effective ``propose-vote-propose-vote'' pattern and
reduces the protocols' inter-proposal time to
$2\Delta$.\footnote{HotStuff~\cite{hotstuff}
also advertises to advance
in rounds matching the actual network delay $\delta$, decoupling progress from
the pessimistic bound $\Delta$; this property, known as \emph{responsiveness},
is shared by several protocols~\cite{pbft,hotstuff-1, hotstuff-2,
jolteon, moonshot}. Since the dominant strategy for profit-maximizing proposers
is to exhaust their timeouts, blockchain protocols de-facto run
non-responsively; we discuss this below and in \cref{sec:relwork}.}
Simplex~\cite{simplex} also achieves a $2\Delta$ inter-proposal time: unlike
the HotStuff family, which attains this bound via chained QCs in a single-vote
pipeline, Simplex advances to the next slot immediately upon notarization of
the current slot, while a distinct finalization vote proceeds in parallel with
the subsequent proposal.

Optimistic designs push pipelining even further: under the assumption of
consecutive honest leaders, Moonshot~\cite{moonshot} and
Hydrangea++~\cite{hydrangeaplusplus} allow leaders to issue new proposals even
before the previous one has been quorum-certified. This aggressive pipelining
results in a ``propose-propose-propose'' structure, enabling a new proposal at
every round and achieving the state-of-the-art minimum inter-proposal time of
$\Delta$. Because transactions may arrive at any point between two consecutive
proposals, the expected waiting time until inclusion in a proposal is half of
the inter-proposal time; the lowest transaction latency achieved by existing
protocols is the $3.5\Delta$ of Moonshot and Hydrangea++.

\myparagraph{Inter-Proposal Time: Stable Leaders}
In classical atomic broadcast protocols, proposals are cryptographically
chained: a leader can successfully propose a new batch of transactions only
after receiving the previous leader's proposal. As a result, when the leader
rotates at every slot, the inter-proposal time must be at least $\Delta$ to
give time to the previous proposal to propagate to the next leader. Protocols
achieving a sub-$\Delta$ inter-proposal time typically allow a single leader to
remain in control for multiple consecutive slots and issue several proposals in
sequence (stable or slowly rotating leaders). In this regime, the leader has
local access to its own prior proposals and can therefore bypass network
propagation delays. This is the case, for instance, of the original
PBFT~\cite{pbft}. Keeping the same leader in power for extended periods is,
however, problematic in blockchains, as it grants the leader disproportionate
control over transaction ordering, exacerbating MEV extraction, censorship, and
eroding fairness among nodes. A middle ground is achieved by some PBFT-like
protocols instantiated with slowly rotating leaders: they allow a single leader
to produce multiple proposals within a fixed number of consecutive slots
(called a window), separated only by local processing delays (a few tens of
milliseconds). When the leader rotates across windows, though, they incur a
larger inter-proposal time of at least $\Delta$ (in practice, a few hundreds of
milliseconds). Another example of such protocols is Alpenglow~\cite{alpenglow}.

\myparagraph{On Responsiveness}
A substantial body of work has been devoted to the design of \emph{responsive}
Byzantine fault-tolerant protocols. Responsiveness refers to the ability of a
protocol to make progress as a function of the actual network conditions
(rounds of length $\delta$), rather than operating with conservative timeouts
that reflect the worst-case network bound ($\Delta$).
The goal is to
reduce latency under $\delta \ll \Delta$ network conditions.

Recent work~\cite{time-is-money,chorusone-twitter} shows that, in blockchains,
the dominant strategy of profit-maximizing proposers is to delay their proposal
until their timeout expires: this allows them to gain more information about
incoming transactions and craft a proposal that maximizes the Maximal
Extractable Value (MEV). This behavior is often referred to as \emph{timing
games}, and has been observed in production systems such as
Ethereum~\cite{time-is-money} and Solana~\cite{chorusone-twitter}. As a result
of economic incentives, blockchain protocols de-facto run non-responsively at
$\Delta$ speed, and not at the $\delta$ speed that responsiveness would allow.
For this reason, in this work, we consider protocol rounds of length $\Delta$.

\subsection{Our Contribution}

\myparagraph{\composite}
\composite (\cref{sec:protocol}) runs $K$ parallel and independent instances of
a closed-box atomic broadcast component protocol---e.g., PBFT, Tendermint,
Simplex, Moonshot, or a DAG-based protocol such as Bullshark~\cite{bullshark},
Shoal~\cite{shoal}, or Mysticeti~\cite{mysticeti}. Each component instance is
parameterized with a fixed, regular proposal schedule that generates a proposal
every $\deltaibt$ time, so when all nodes are honest each instance extends its
log every $\deltaibt$. In \composite, $K$ is a tunable parameter, and the
proposal schedules of consecutive instances are offset by $\deltaibt/K$. As
illustrated in \cref{fig:intro} for $K=4$, the component protocols together
produce a $\deltaibt/K$ inter-proposal time. By appropriately choosing $K$,
\composite achieves an arbitrarily small inter-proposal time.

\begin{figure*}[t]
    \centering
    \includegraphics[width=\linewidth]{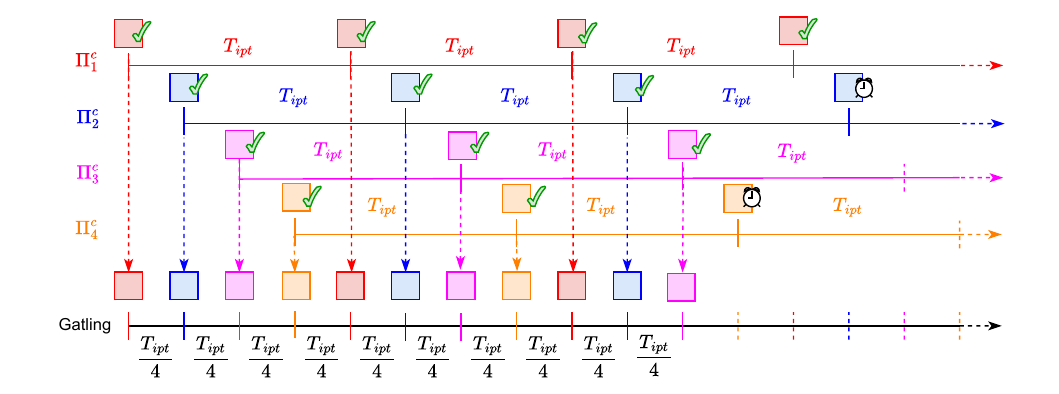}
    \caption{The \composite protocol depicted in this figure runs $K = 4$
    instances of a component protocol $\PIcomponent$. The \composite output
    log is formed by deterministically interleaving blocks from all instances,
    in increasing order of slot number and instance number. A block is
    confirmed once all preceding blocks in this global order are confirmed.
    Checkmarks indicate blocks that are confirmed, whereas clocks indicate
    blocks that are not yet confirmed. The last red block is not yet included
    in the \composite log because, in the interleaved ordering, it comes after
    the third orange block, which has not yet been confirmed. This does not
    happen in the good-case but it could happen in case of, e.g., a network
    hiccup.}
    \label{fig:intro}
\end{figure*}

Since we want \composite to output a single transaction log rather than $K$
independent parallel logs, the $K$ component logs are combined via a
deterministic merge rule. All nodes assign the same index from $1$ to $K$ to
each component instance, and each node locally constructs the \composite output
log by interleaving the $K$ component logs in order of increasing slot number
and instance number. For example, when $K=2$, the \composite log alternates
slot by slot: the slot-$1$ block of instance $1$, then the slot-$1$
block of instance $2$, then the slot-$2$ block of instance $1$, and so on.
\composite confirms a block once all slots that come before it have been
confirmed.

We prove that \composite implements secure atomic broadcast with transaction
latency $(\deltaibt/K, \deltaconf)$, where $\deltaibt/K$ is the inter-proposal
time---which can be made arbitrarily small (\cref{sec:analysis})---and
$\deltaconf$ is the confirmation time. We perform a stochastic latency analysis (\cref{sec:tail-latency}) to characterize the impact of occasional proposal confirmation failures and identify the optimal inter-proposal time that minimizes expected end-to-end latency.
We implement a \composite proof of
concept on top of a production-grade Simplex
deployment~\cite{simplex,shoup-simplex,commonware-alto} with a deliberately
unoptimized $500$\,ms per-instance inter-proposal time, and run it on a global
cluster with nodes in Asia, Europe, and North America. End-to-end transaction
latency falls as $K$ grows, reaching $244$\,ms at $K=9$ under a $1\%$
proposal-drop rate while driving the slot time down to $50$\,ms
(\cref{sec:evaluation}).

\myparagraph{Discussion and Related Work}
We further discuss how the design of vanilla \composite inherently loses what
is sometimes called \emph{predictable
validity}~\cite{10.1145/3558535.3559777, 10.1007/978-3-031-18283-9_27,
Kiffer2023NakamotoCU} (\cref{sec:tx-deduplication})---a property that enables proposers to reason at
proposal time about transaction execution outcomes---and we propose two
variants that retain it (\cref{sec:predictable-validity}) under a slowly
rotating leader schedule. The first variant retains predictable validity by
mimicking the behaviour of, e.g., PBFT-style protocols. The second retains
predictable validity at the execution level by creating two distinct tiers of
block space and interleaving their execution
(\cref{sec:predictable-validity}). One tier, named \emph{subprime block space},
is reserved for transactions that are not time-critical and cannot be
front-run (e.g., layer-2 data blob transactions); the other, \emph{prime block
space}, is reserved for time-critical and highly contentious transactions
(e.g., DeFi).

A long line of work has sought to reduce consensus latency through the design
of responsive protocols complemented with
pipelining~\cite{hotstuff,hotstuff-1,hotstuff-2,jolteon,moonshot},
speculative execution~\cite{hotstuff-1}, optimistic
assumptions~\cite{moonshot}, or fast-path techniques that trade resilience for
fewer voting rounds~\cite{minimmit,alpenglow,kudzu,hydrangea}. We compare
\composite with these approaches in more detail in \cref{sec:relwork}.
\composite improves latency along an orthogonal dimension: it achieves
arbitrarily small inter-proposal times through composition, while each
component protocol operates under conservative worst-case network assumptions
and traditional Byzantine fault-tolerance thresholds. Notably,
\emph{\composite shows that even starting from a simple, unoptimized protocol
with poor standalone performance}---e.g., Simplex with a $500$\,ms
inter-proposal time---\emph{parallel composition alone is sufficient to
achieve latencies improving upon those of highly optimized protocols}. This
highlights that \emph{extreme protocol-level optimization is not a
prerequisite for ultra-low latency}. Moreover, \composite is designed to
remain compatible with responsive, pipelined, speculative, and fast-path
techniques implemented within its component protocols, enabling these
mechanisms to be seamlessly layered on top of it.

Conceptually, \composite can be viewed as extending the parallelism underlying DAG-based protocols~\cite{bbca,dag-rider,bullshark,tusk,shoal,mysticeti}. Whereas DAG protocols support multiple concurrent ordering streams over a shared graph of proposals that reference one another, \composite removes these references entirely by composing independent consensus instances and imposing staggered proposal schedules. The \composite log is recovered with a deterministic merge rule that does not need proposals to reference one another, thereby allowing the inter-proposal time to drop below the network delay.

After completion of this work, it was brought to our attention that
Shoal++~\cite{shoalplusplus} employs a construction similar to \composite in
the context of DAG-based protocols to reduce the inter-proposal time. In
particular, it operates three parallel DAG instances, time-staggered by a
fixed $\Delta$ offset, in a not-fully-generic ``ajar-box'' manner. By contrast,
\composite explores regimes with inter-proposal times below the network delay
and provides variants that retain predictable validity even when the offsets
become arbitrarily small. \composite is closed-box, and our experiments allow
us to isolate the latency reduction achieved purely from time-staggered
parallel composition.

\section{Model and Preliminaries}
\label{sec:model}

Our overarching goal is to design an atomic broadcast protocol.
We consider a setting with $n$ nodes, each of which equipped with a cryptographic identity (typically, a public/secret key pair for digital signatures) known to all other nodes (public key infrastructure, PKI). %
Over time, the nodes receive transactions from the \emph{environment} and, over time, they return a log of transactions to the environment.
The nodes send each other messages through a \emph{network} to reach agreement on an ordering of their input transactions into their output logs. The nodes' output logs are \emph{confirmed}.
Time proceeds in discrete \emph{rounds}. Throughout the entire protocol execution, the nodes have \emph{synchronized clocks}: all nodes always know which is the current round.
In each round, each node receives messages from the network and possibly transactions from the environment. In every round, each node updates its internal state, produces messages to send to other nodes via the network, and outputs a log of confirmed transactions.

An \emph{adversary} aims to disrupt the system by corrupting nodes and delaying network messages. We consider an adversary that corrupts a fixed set of at most $f$ nodes at the start of the protocol execution, before any randomness is drawn (static corruption). The adversary learns the internal state of corrupted nodes and causes them to deviate arbitrarily from the protocol for the entire execution (permanent Byzantine faults). Corrupted nodes, also known as adversarial nodes, run any algorithm chosen by the adversary in time polynomial in the size of the inputs (computationally bounded adversary).

Nodes can exchange messages between each other via a fully-connected network of point-to-point links. The network is \emph{partially synchronous}~\cite{10.1145/42282.42283}: there exits an upper-bound message delay $\Delta$ that is known to all nodes and a global stabilization time $\GST<\infty$ that is adaptively chosen by the adversary. Before $\GST$, the adversary can delay messages arbitrarily (asynchronous network). After $\GST$, the adversary can delay messages by at most $\Delta$ (synchronous network). During synchronous periods, when a node instructs the network to send a message to another node, the message is enqueued in the recipient's pending message queue together with a countdown initialized to $\Delta$; the countdown decreases by 1 with each round.  The adversary can also decrease the countdown of every message at will.
Once the countdown hits 0, the message is delivered to the recipient at the
beginning of the next round. This means that every
message sent by an honest node by time $t$ is delivered to all honest nodes by time $\max(t, \GST) + \Delta$.

\myparagraph{Atomic Broadcast Interface} The nodes run a Byzantine fault-tolerant consensus protocol of the atomic broadcast flavour to reach agreement on their output logs, i.e., a common ordering of their input transactions. %

\begin{definition}[Atomic Broadcast Interface]
    \label{def:ab-interface}
    An atomic broadcast protocol is run by $n$ nodes.
    In every round, nodes possibly receive transactions from the environment and they return to the environment a confirmed sequence of transactions called output log.
\end{definition}

We consider atomic broadcast executions in which, regardless of the strategy of the  adversary, all correct nodes output a log $\LOG$ that is safe and live. We write $\LOG \preceq \tilde{\LOG}$ to denote that $\LOG$ is a prefix of $\tilde{\LOG}$. We denote by $\bcast$ the function that a node invokes to broadcast a transaction to the network, and by $\mathsf{deliveredLog}$ the function that a node invokes to retrieve its output log.

\begin{definition}[Atomic Broadcast Security]
\label{def:ab-security}
An atomic broadcast protocol is \emph{secure with resilience $\tau$} in partially synchronous networks iff,
in every partially-synchronous execution with up to $f \leq \tau$ adversarial nodes,
except with probability $\operatorname{negl}(\kappa)$ where $\kappa$ is the security parameter,
the following properties hold:
    \begin{itemize}
        \item \textbf{Safety}:
        For every two honest nodes $p,q$ and for every two times $t,t'$, if $\LOG$ is the log output by $p$ at $t$ and $\tilde\LOG$ is the log output by $q$ at $t'$, then the two logs are consistent, i.e., $\LOG \preceq \tilde\LOG$ or vice versa.

        \item \textbf{Liveness}:
        For every $t_0$ and every transaction $\tx$, if every honest node $p$ calls $\bcast$ with parameter $\tx$ by time $t_0$,
        then there exists a $t_1 \geq t_0$ such that
        for every $t_2 \geq t_1$ and for every
        $p$ that calls $\mathsf{deliveredLog}$ at $t_2$, $\tx \in \LOG$.
    \end{itemize}
\end{definition}

\myparagraph{Enriched Atomic Broadcast Interface}
As components of \composite, we need an atomic broadcast protocol that proceeds in slots.
Each \emph{slot} is associated with a leader node and a fixed proposal time $\Tproposal$, during which the leader is expected to propose a payload of transactions to be appended to the output log, and broadcast it to the network.
Slots typically consists of three phases: a proposal phase, and two voting phases known as notarization and confirmation. Each payload that is included in the output log is annotated with the slot number in which it was proposed. The slot number is non-decreasing along the log entries. %

We additionally require the atomic broadcast protocol to implement an enriched interface. This interface is obtained by parameterizing the protocol with a sequence of proposal times, denoted $\Tproposal$, and by enriching the classical atomic broadcast interface with two additional functions, $\bcast$ and $\deliveredLog$. Since the \composite protocol runs several parallel instances of an atomic broadcast protocol, the two new functions make explicit when \composite writes an input to an instance, or reads the output from an instance.
For clarity, we now introduce the enriched interface of the atomic broadcast protocol. %

\begin{definition}[Enriched Atomic Broadcast Interface]
    \label{def:ab-enriched-interface}
    An enriched atomic broadcast protocol $\Pi$ is run by $n$ nodes, and parameterized by a sequence of proposal times $\Tproposal$. %
    It provides a function $\bcast$ that a node can call with a transaction $\tx$ as parameter and that returns nothing. It also provides a function $\deliveredLog$
    which a node can call without any parameter and that returns an annotated output log, i.e., a sequence of transactions $\LOG$ annotated with the slot number in which they were proposed.
\end{definition}

We emphasize that the enriched interface is merely a syntactic variant of the standard atomic broadcast interface, and does not alter its security semantics.
We highlight that virtually all partially synchronous atomic broadcast protocols we are aware of can be straightforwardly adapted to provide the enriched interface. %
These protocols already operate over a notion of slot---called \emph{view} in PBFT~\cite{pbft}, Tendermint~\cite{buchman2016tendermint}, HotStuff~\cite{hotstuff}, and Moonshot~\cite{moonshot}, \emph{heights} in Simplex~\cite{simplex}, and \emph{epochs} in Streamlet~\cite{streamlet}. They can be easily modified to enforce a fixed, regular schedule of slot proposal times (e.g., every $2\Delta$ time) and to annotate confirmed payloads with the slot number in which they were proposed, a value over which the protocols already establish consensus.

Finally, we are interested in atomic broadcast protocols that fulfill a stronger notion of liveness, which can be articulated into two properties. We name the first property \emph{slot-driven liveness}.

\begin{definition}[Slot-Driven Liveness]
    \label{def:slot-driven-liveness}
    An atomic broadcast protocol has \emph{slot-driven liveness} iff, for any slot whose associated leader is honest, a (possibly empty) proposal is produced and becomes confirmed by all honest nodes within $3\Delta$ time from the slot's start.
\end{definition}

The second property is \emph{good-case $(\deltaibt,\deltaconf)$-latency}, with $\deltaconf$ being the time it takes for a proposal to be confirmed after it has been proposed. Before stating this property, we need to define what it means for a sequence of proposal times to be $\deltaibt$-spaced.

\begin{definition}[$\deltaibt$-Spaced Sequence]
    \label{def:deltaibt-spaced-apart}
        A sequence $T$ %
        is \emph{$\deltaibt$-spaced} iff, for every two distinct $t_1,t_2 \in T$, %
        $| t_2 - t_1 | \geq \deltaibt$.
    \end{definition}

Informally, a protocol has good-case $(\deltaibt,\deltaconf)$-latency if, when all nodes are honest, transactions are confirmed $\deltaibt + \deltaconf$ time after they are received by all nodes.

\begin{definition}[Good-Case $(\deltaibt,\deltaconf)$-Latency]
\label{def:deltaibt-deltaconf-good}
    A protocol
    has good-case $(\deltaibt,\deltaconf)$-latency iff, when the protocol is instantiated with a $\deltaibt$-spaced proposal sequence $\Tproposal$, if $\GST=0$ and all nodes are honest,
    then for every $t \in \Tproposal$ and for every transaction $\tx$, if $\tx$ is received by every honest node
    by $t$, then $\tx$ is confirmed by every honest node by $t+ \deltaconf$.
\end{definition}

\section{Protocol}
\label{sec:protocol}

\begin{figure*}[t]
    \centering
            \includegraphics[width=\linewidth]{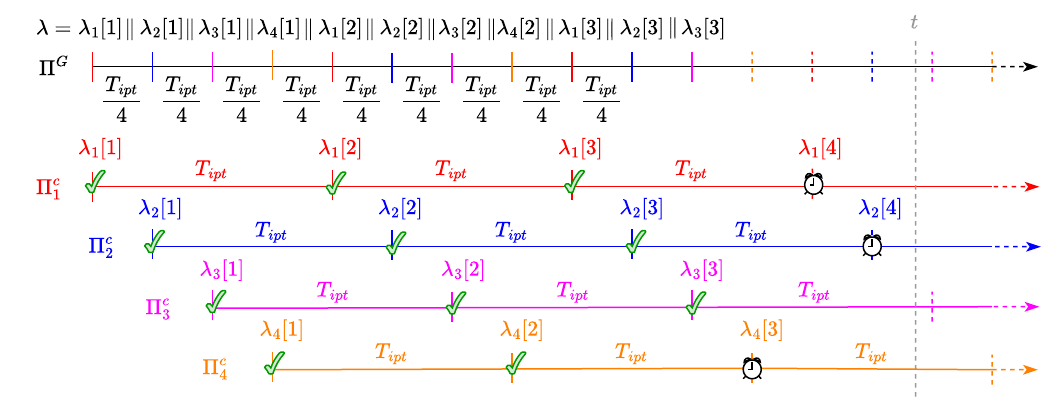}
        \caption{%
        The \composite protocol $\PIcomposite$ depicted in this figure runs $K = 4$ instances of a component protocol $\PIcomponent$.
        The union of the proposal sequences of $\PIcomponent_1$ (red), $\PIcomponent_2$ (blue), $\PIcomponent_3$ (magenta), $\PIcomponent_4$ (orange) is the $\deltaibt/4$-spaced proposal sequence of $\PIcomposite$.  We denote by $\LOG_j[i]$ the payload confirmed in the $i$-th slot of the $j$-th instance. The \composite log $\LOG$ at time $t \geq \GST$ is reconstructed at the top as the concatenation of the payloads of the component protocols, ordered by increasing slot and increasing instance number. Checkmarks indicate blocks that are confirmed, whereas clocks indicate blocks that are not yet confirmed.  Note that $\LOG_4[3],\LOG_1[4],\LOG_2[4],$ etc, are not yet confirmed by \composite (not yet part of $\LOG$) because, at time $t$, the corresponding component protocols have not yet confirmed the payloads for those slots.}
        \label{fig:staggering}
\end{figure*}

The goal of this work is to design an atomic broadcast protocol with an arbitrarily small inter-proposal time---in particular, smaller than the bound on the network delay. To this end, we design \composite with a compositional approach. We describe the protocol step by step and we provide the pseudocode executed by the nodes in \cref{alg:composite-prot-optimized}.

The \composite protocol is parameterized by a positive interger $K>0$, a positive integer $\deltaibt>0$, and a black-box atomic broadcast protocol $\PIcomponent$ that exposes the enriched interface of \cref{def:ab-enriched-interface} and satisfies the slot-driven liveness and good-case (\deltaibt,\deltaconf)-latency of \cref{def:slot-driven-liveness,def:deltaibt-deltaconf-good}. For simplicity, we refer to $\PIcomponent$ as the component protocol. %
Nodes running \composite execute $K$ concurrent and mutually independent instances of the component protocol.  Nodes uniquely identify these instances by assigning them indices from 1 to $K$; in particular, all nodes assign the same index to the same protocol instance. Each instance maintains its own leader schedule, produces its own annotated output log, and is parameterized by a proposal schedule that generates a proposal every $\deltaibt$ time. The proposal schedule of each instance is offset by $\deltaibt/K$ relative to the previous one (\alglocref{alg:composite-prot-optimized}{line:instances-component-prot}). Consequently, the union of the proposal sequences of all $K$ component protocols forms a single proposal sequence with inter-proposal time $\deltaibt/K$.

When a node receives a transaction $\tx$ from the environment, it invokes the $\bcast$ function (\alglocref{alg:composite-prot-optimized}{lst:line:broadcast-composite}) that forwards $\tx$ to all $K$ instances of the component protocol (\alglocref{alg:composite-prot-optimized}{line:for-loop-k-instances-2,lst:line:broadcast-component}). Then, $\tx$ is picked up by the leaders of the instances and included in a proposal. In particular, once the transaction has been received by all nodes, the instance proposing next has a proposal at most $\deltaibt/K$ time away. Therefore, the transaction is guaranteed to appear in some proposal within at most $\deltaibt/K$ time and confirmed after an additional $\deltaconf$ time by the corresponding protocol instance.

\begin{algorithm}[t]
  \caption{Algorithm run by a node of the \composite protocol. }
  \label{alg:composite-prot-optimized}
    \begin{algorithmic}[1]
        \RealFunction{new $\PIcomposite(\PIcomponent, K, \deltaibt)$}
            \LineComment{Start component protocols with proposal schedules s.t. the \composite inter-proposal time is $\deltaibt/K$}
            \For{$k=1,\dots,K$} \label{line:for-loop-k-instances}
                \State $\pi_k \gets \text{new } \PIcomponent(\Tproposal_k = \{ (s+\frac{k-1}{K})\cdot\deltaibt \mid s \in \IN_0 \})$ \label{line:instances-component-prot}
            \EndFor
        \EndRealFunction
        \smallskip
        \RealFunction{$\PIcomposite.\operatorname{deliveredLog}()$}  \label{line:delivered-log}
            \LineComment{Get component protocols' output logs annotated with slot number }
            \For{$k=1,\dots,K$}
                \State $\LOG_{k} \gets \pi_k.\operatorname{deliveredLogAnnotated}()$ \label{line:annotated-logs}
            \EndFor
            \LineComment{Find max slot number for which all component protocols have decided}
            \State $s^* \gets \text{max } s \in \IN_0 \text{ so that } \forall s \leq s^*: \forall k \leq K: \LOG_{k}[s] \neq \bot$ \label{line:v-star}

            \LineComment{Find max instance number that has already decided for slot $s^*+1$}
            \State $k^* \gets \text{max } k \in \IN_0 \text{ so that } \forall k \leq k^*: \LOG_{k}[s^*+1] \neq \bot$ \label{line:k-star}
            \LineComment{Output \composite protocol's current log}
            \If{$k^* = 0$} \label{line:k-star-0}
                \State\Return{$\LOG_{1}[1] \| \dots \| \LOG_{k}[1] \| \dots \| \LOG_{1}[s^*] \| \dots \| \LOG_K[s^*]$} \label{lst:line:return1}
            \Else
                \State\Return{\mbox{$\LOG_{1}[1] \| \dots \| \LOG_{k}[1] \| \dots \| \LOG_{1}[s^*+1] \| \dots \| \LOG_{k^*}[s^*+1]$}} \label{lst:line:return2}
            \EndIf
        \EndRealFunction
        \smallskip
        \RealFunction{$\PIcomposite.\operatorname{broadcast}(\tx)$} \label{lst:line:broadcast-composite}
            \LineComment{Forward \composite protocol's inputs to component protocols}
            \For{$k=1,\dots,K$} \label{line:for-loop-k-instances-2}
                \State $\pi_k.\operatorname{broadcast}(\tx)$ \label{lst:line:broadcast-component}
            \EndFor
        \EndRealFunction
    \end{algorithmic}
\end{algorithm}

\begin{figure*}[t]
    \centering
            \includegraphics[width=\linewidth]{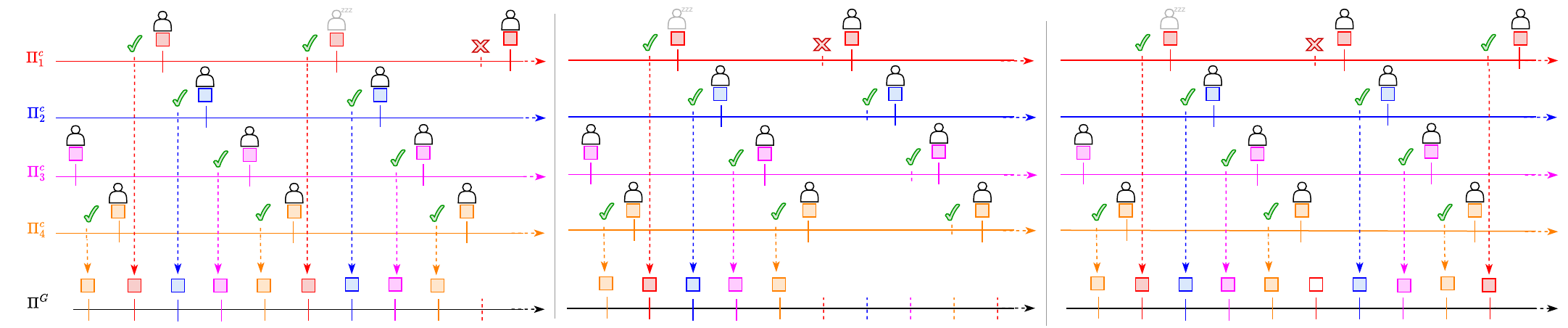}
        \caption{%
        Behavior of \composite in the presence of an adversarial proposer.
        Left: Each component instance assigns a leader per slot; confirmed blocks (green checkmarks) are immediately appended to the \composite log (bottom), if all leaders are honest. If a leader is adversarial (third leader of the red instance, semitransparent), no block is confirmed for that slot (red cross). Center: The other instances (blue, magenta, orange) continue confirming blocks independently, but \composite delays including them in the log because it waits for the red instance to reach consensus on the adversarial slot. Right: Once the red instance progresses under a subsequent honest leader and confirms a proposal, \composite resumes extending the log, appending an empty block for the adversarial slot together with the pending confirmed blocks from the other instances.}
        \label{fig:faulty-leader}
\end{figure*}

It is not sufficient for $\tx$ to be confirmed by a component protocol: we need it to be confirmed by the \composite protocol. Towards this, we describe how the output log of \composite is built, and visually illustrate this in \cref{fig:staggering}.
The nodes construct the \composite output log by concatenating the $K$ output logs of the component protocols as follows: For each slot index $i$, the nodes define the $i$-th slot of the \composite protocol as the tuple $s_i = (s_{i,1}, \dots, s_{i,K})$, where $s_{i,j}$ denotes the $i$-th slot of the $j$-th instance. Then, they construct the \composite output log as the sequence of payloads confirmed in the \composite slots. In other words, the \composite output log concatenates the payloads of the component protocols output logs, ordered by increasing slot and instance number. In \composite, proposals are not chained. With this in mind, we now define the \composite confirmation rule, which consists of three conditions. The payload of slot $s_{i,j}$ is confirmed if: (1) instance $j$ has confirmed the payload of slot $s_i$, (2) all instances $j = 1, \dots, K$ have decided for all slots $s_{i'}$ with $i' < i$, and (3) all instances $j'$ with $j' < j$ have decided for slot $s_i$. As shown in \cref{fig:faulty-leader}, a component protocol that fails to confirm a proposal for a slot (e.g., due to an adversarial leader) still reaches consensus on the slot being empty. %

In the next section, we prove that the \composite protocol in \cref{alg:composite-prot-optimized} is secure, and it has a good-case $(\deltaibt/K,\deltaconf)$-latency.

\section{Analysis}
\label{sec:analysis}

We now state the main results of this work, namely the security and the good-case $(\frac{\deltaibt}{K}, \deltaconf)$-latency of \composite. Proofs are deferred to the appendix due to space constraints.

\begin{theorem}[\composite Security]
\label{thm:composite-security}
    If $\, \PIcomponent$ is a secure atomic broadcast protocol as per \cref{def:ab-security}, then the \composite protocol in \cref{alg:composite-prot-optimized} that uses $\PIcomponent$ as its component protocol is also secure as per \cref{def:ab-security}.
\end{theorem}

\deferredproof{proof}{thm:composite-security}

\begin{defer}{proof}
\begin{proof}[Proof of \cref{thm:composite-security}]
    \myparagraph{Safety}
    For any
    two arbitrary honest nodes $p$ and $q$, and any two arbitrary points in time $t$ and $\tilde t$, let $\LOG$ be the log output by $p$ at time $t$ and $\tilde\LOG$ the log output by $q$ at time $\tilde t$.
    We will show that $\LOG$ and $\tilde\LOG$ are consistent, i.e., $\LOG \preceq \tilde\LOG$ or $\tilde\LOG \prec \LOG$.

    We denote with $s^*$ the maximum slot number for which all component protocols have decided, and with $k^*$ the maximum instance number that has already decided for slot $s^* +1$ (as per \alglocref{alg:composite-prot-optimized}{line:v-star,line:k-star}).
    Let $s^*_p$ and $k^*_p$ denote the values of $s^*$ and $k^*$ observed by $p$ at time $t$ and $s^*_{q}$ and $k^*_{q}$ denote the values of $s^*$ and $k^*$ observed by $q$ at time $\tilde t$. %
                Suppose  $k^*_p = 0 \land k^*_q = 0$, i.e., $\LOG$ and $\tilde\LOG$ are both returned by \alglocref{alg:composite-prot-optimized}{lst:line:return1}. %
                Therefore: $\LOG = \LOG_{1}[1] \| \dots \| \LOG_{K}[1] \| \dots \| \LOG_{1}[s_p^*] \| \dots \| \LOG_{K}[s_p^*]$ and $\tilde\LOG = \tilde\LOG_{1}[1] \| \dots \| \\ \tilde\LOG_{K}[1] \| \dots \| \tilde\LOG_{1}[s_q^*] \| \dots \| \tilde\LOG_K[s_q^*]$.
                For every $k \in \{1, \dots, K \}$, %
                the annotated output logs
                $\LOG_{k}$ and $\tilde \LOG_{k}$
                composing $\LOG$ and $\tilde \LOG$, respectively, are the output logs of instances of
                $\PIcomponent$ (\alglocref{alg:composite-prot-optimized}{line:annotated-logs}). Because $\PIcomponent$ is a secure atomic broadcast protocol, the safety property of $\PIcomponent$ ensures that the annotated output logs are consistent.
                Therefore,
                $\forall s \leq \min(s^*_p,s^*_q),
                \forall k \in \{1, \dots, K \} : \LOG_{k}[s] = \tilde \LOG_{k}[s]$. Hence, $\LOG$ and $\tilde\LOG$ are consistent logs.
                Without loss of generality, suppose $k^*_p = 0 \land k^*_q > 0$, i.e., $\LOG$ and $\tilde\LOG$ are returned by \alglocref{alg:composite-prot-optimized}{lst:line:return1,lst:line:return2}, respectively. The case $k^*_p > 0 \land k^*_q = 0$ proceeds analogously. %
                Therefore: $\LOG = \LOG_{1}[1] \| \dots \| \LOG_{K}[1] \| \dots \| \LOG_{1}[s^*_p] \| \dots \| \LOG_{K}[s^*_p]$ and $\tilde\LOG = \tilde\LOG_{1}[1] \| \dots \| \\ \tilde\LOG_{K}[1] \| \dots \| \tilde\LOG_{1}[s_q^*] \|  \dots \| \tilde\LOG_K[s_q^*]\| \tilde\LOG_{1}[s_q^*+1] \| \dots \| \tilde\LOG_{k_q^*}[s_q^*+1]$. For every $k \in \{1, \dots, K \}$, %
                the annotated output logs
                $\LOG_{k}$ and $\tilde \LOG_{k}$
                composing $\LOG$ and $\tilde \LOG$, respectively, are the output logs of instances of
                $\PIcomponent$ (\alglocref{alg:composite-prot-optimized}{line:annotated-logs}). Because $\PIcomponent$ is a secure atomic broadcast protocol, the safety property of $\PIcomponent$ ensures that the annotated output logs are consistent.
                Consider the following two cases: $s^*_p \geq s^*_q+1$ and $s^*_q \geq s^*_p$.
                If $s^*_p \geq s^*_q+1$, then for every $k \in \{1, \dots, K \}$, and for every $s \leq s^*_q$,
                $\LOG_{k}[s] = \tilde \LOG_{k}[s]$. Additionally, for every $j \leq k^*_q$,
                $\LOG_{j}[s^*_q+1] = \tilde \LOG_{j}[s^*_q+1]$.  Therefore, $\LOG$ and $\tilde\LOG$ are consistent logs.
                If $s^*_q \geq s^*_p$, then for every $k \in \{1, \dots, K \}$, and for every $s \leq s^*_p$,
                $\LOG_{k}[s] = \tilde \LOG_{k}[s]$. %
                Therefore, $\LOG$ and $\tilde\LOG$ are consistent logs. %

                Suppose $k^*_p > 0 \land k^*_q > 0$.
                Because $k^*_p > 0 \land k^*_q > 0$, $\LOG$ and $\tilde\LOG$ are both returned by \alglocref{alg:composite-prot-optimized}{lst:line:return2}. Therefore: $\LOG = \LOG_{1}[1] \| \dots \| \LOG_{K}[1] \| \dots \| \\ \LOG_{1}[s^*_p] \| \dots \| \LOG_{K}[s^*_p] \| \LOG_{1}[s_p^*+1] \| \dots \| \LOG_{k_p^*}[s_p^*+1]$ and
                $\tilde\LOG = \tilde\LOG_{1}[1] \| \dots \| \\ \tilde\LOG_{K}[1] \| \dots \| \tilde\LOG_{1}[s_q^*] \| \dots \| \tilde\LOG_K[s_q^*] \| \tilde\LOG_{1}[s_q^*+1] \| \dots \| \tilde\LOG_{k_q^*}[s_q^*+1]$.
                For every $k \in \{1, \dots, K \}$, %
                the annotated output logs
                $\LOG_{k}$ and $\tilde \LOG_{k}$
                composing $\LOG$ and $\tilde \LOG$, respectively, are the output logs of instances of
                $\PIcomponent$ (\alglocref{alg:composite-prot-optimized}{line:annotated-logs}). Because $\PIcomponent$ is a secure atomic broadcast protocol, the safety property of $\PIcomponent$ ensures that the annotated output logs are consistent.
                Without loss of generality, suppose $s^*_p \geq s^*_q$. The case $s^*_p < s^*_q$ proceeds analogously. Consider the following cases: $k^*_p \geq k^*_q$ and $k^*_p < k^*_q$. If $k^*_p \geq k^*_q$, then
                for all $k \in \{1, \dots, K \}$,
                for all $j \leq k^*_q$, and
                for all $s \leq s^*_q$, we have that
                $\LOG_{k}[s] = \tilde \LOG_{k}[s]$
                and $\LOG_{j}[s^*_q+1] = \tilde \LOG_{j}[s^*_q+1]$.
                If $k^*_p < k^*_q$, then
                for all $k \in \{1, \dots, K \}$,
                for all $j \leq k^*_p$, and
                for all $s \leq s^*_q$, we have that
                $\LOG_{k}[s] = \tilde \LOG_{k}[s]$
                and $\LOG_{j}[s^*_q+1] = \tilde \LOG_{j}[s^*_q+1]$. Hence, $\LOG$ and $\tilde\LOG$ are consistent logs.

    \myparagraph{Liveness}
    For every node $p$ and time $t$, we denote by $\AT{t}{p}{\LOG}$ the \composite log output by $p$ at $t$ (the output of $\mathsf{deliveredLog}$ (\alglocref{alg:composite-prot-optimized}{line:delivered-log})). Similarly, we denote by $\AT{t}{p}{\LOG_k}$ the log of the component protocol that is output by $p$ at $t$ by instance $k$ (the output of $\deliveredLog$ (\alglocref{alg:composite-prot-optimized}{line:annotated-logs})).

    We show that for every $t_0$ and every transaction $\tx$, if every honest node $p$ calls the \composite $\bcast$ function with parameter $\tx$ by time $t_0$, then there exists a $t_1 \geq t_0$ such that for every $t_2 \geq t_1$ and for every $p$ that calls $\mathsf{deliveredLog}$ at $t_2$ to obtain $\AT{t_2}{p}{\LOG}$,
    we have that $\tx \in \AT{t_2}{p}{\LOG}$.
    Because every honest node runs the \composite protocol described in \cref{alg:composite-prot-optimized}, it runs $K$ parallel instances of the component protocol, as per \alglocref{alg:composite-prot-optimized}{line:instances-component-prot}. Because of \alglocref{alg:composite-prot-optimized}{lst:line:broadcast-component}, we know that every honest node calls the $\bcast$ function of every instance of the component protocol with parameter $\tx$  by $t_0$ (\alglocref{alg:composite-prot-optimized}{lst:line:broadcast-composite}).
    From the liveness property of the component protocol, we know that, for every instance $k \in \{ 1, \dots, K \}$, there exists a $t_{1} \geq t_{0}$
    such that for every $t_{2} \geq t_{1}$ and for every $p$ that calls $\deliveredLog$
    (\alglocref{alg:composite-prot-optimized}{line:annotated-logs}) at $t_{2}$ it holds that $\tx \in \AT{t_{2}}{p}{\LOG_k}$. Let $t_{1,k}$ be the time $t_1$ for instance $k$.
    Let $t' = \min(t_{1,1}, \dots, t_{1,K})$, and let $j$ be the instance for which $t_{1,j} = t'$. In other words, $j$ is the first instance for which every honest node $p$ observes $\tx \in \AT{}{p}{\LOG_j}$.
    Let $s^{**}$ be the slot of instance $j$ associated with a payload that includes $\tx$.
    For the payload of $s^{**}$ to be included in the \composite log returned by \alglocref{alg:composite-prot-optimized}{lst:line:return1,lst:line:return2}, it is required that for every slot $s < s^{**}$, all instances $k \leq j$ deliver an annotated output log with slot number $s$, and for slot $s = s^{**}$, all instances $k < j$ deliver an annotated output log with slot number $s$ (\alglocref{alg:composite-prot-optimized}{line:v-star,line:k-star}).
    Because the component protocols satisfy the slot-driven liveness property and they are parameterized by an ever growing $\deltaibt$-spaced proposal sequence, honest nodes will always produce a proposal for the slots they are assigned to. Therefore, at every slot with an honest proposer, the annotated output log of every instance grows and the maximum slot number for which all component protocols have decided  ($s^*$  in \alglocref{alg:composite-prot-optimized}{line:v-star}) increases.
    Hence, all instances $k \leq j$ eventually deliver an annotated output log with slot $s$ or higher. Because of the \composite confirmation rule, the \composite log delivers $\AT{}{}{\LOG}_j[s^{**}]$ (i.e., the annotated output log of instance $j$ for slot $s^{**}$) only after all instances $k \leq j$ have decided for slot $s^{**}$.
    Thus, for every $t_2 \geq t_1 = \max(t_{1,1}, \dots, t_{1,j})$, every honest party $p$ will observe $\tx \in \AT{t_2}{p}{\LOG}[s^{**}]$. %
    This concludes the proof.
\end{proof}
\end{defer}

    \begin{theorem}[\composite Good-Case $(\frac{\deltaibt}{K},\deltaconf)$-Latency]
\label{thm:composite-prot-latency}
    If $\, \PIcomponent$ has good-case $(\deltaibt,\deltaconf)$-latency, then the composite protocol in \cref{alg:composite-prot-optimized} has good-case $(\frac{\deltaibt}{K},\deltaconf)$-latency.
\end{theorem}

\deferredproof{proof}{thm:composite-prot-latency}

\begin{defer}{proof}
\begin{proof}[Proof of \cref{thm:composite-prot-latency}]
    To prove that the protocol in \cref{alg:composite-prot-optimized} has good-case $(\frac{\deltaibt}{K},\deltaconf)$-latency, we prove that when the protocol is instantiated with a $\frac{\deltaibt}{K}$-spaced proposal sequence, if $\GST = 0$ and all nodes are honest, then for every $t$ in its proposal sequence, and for every transaction $\tx$, if $\tx$ is received by every honest node by $t$, then $\tx$ is confirmed by every honest node by $t+ \deltaconf$.
    Because every honest node runs \cref{alg:composite-prot-optimized}, it runs $K$ parallel instances of a component protocol,
    indexed from 1 to $K$, such that every instance $k \in K$ is parameterized by a proposal sequence $\Tproposal_k = \{ (s+\frac{k-1}{K})\cdot\deltaibt \mid s \in \IN_0 \}$ (\alglocref{alg:composite-prot-optimized}{line:instances-component-prot}). Therefore, the \composite proposal sequence is $\Tproposal = \bigcup_{k = 1}^{K} \Tproposal_k$ and there is one instance proposing every $\frac{\deltaibt}{K}$ time. %
    Because all nodes are honest, if they all receive a transaction $\tx$ by time $t$, with $t \in \Tproposal$, then by time $t$ they call the $\bcast$ function of the \composite protocol with $\tx$ as a parameter (\alglocref{alg:composite-prot-optimized}{lst:line:broadcast-composite}). Therefore, by time $t$, they forward $\tx$ to every instance of the component protocol (\alglocref{alg:composite-prot-optimized}{lst:line:broadcast-component}).
    Because $\Tproposal = \bigcup_{k = 1}^{K} \Tproposal_k$, there exists an instance $k'$ such that $t \in \Tproposal_{k'}$. Because every instance of the component protocol has good-case $(\deltaibt,\deltaconf)$-latency, then every honest node of instance $k'$ will confirm $\tx$ by $t + \deltaibt$. Let $s'$ be the slot in which instance $k'$ confirms $\tx$.
    Because all nodes are honest and $\GST = 0$, and because all instances $1,\dots,k'-1$ of the component protocol have a proposal time for slot $s'$ that precedes $t$, all instances  $1,\dots,k'-1$ will confirm for slot $s'$ before instance $k'$ confirms for slot $s'$. Therefore, all instances $1,\dots,k'-1$ of the component protocol have already decided for slot $s'$ by $t + \deltaconf$. Therefore, $\LOG_k[s']$ can be appended to the output log of the \composite protocol and $\tx$ is confirmed by the \composite protocol by $t +  \deltaconf$. Hence, \composite has a good-case $(\frac{\deltaibt}{K},\deltaconf)$-latency.
\end{proof}
\end{defer}

\section{Optimal Inter-Proposal Time}
\label{sec:tail-latency}

In the previous sections we established that in the good-case, after $\GST$, \composite achieves $(\deltaibt/K, \deltaconf)$-latency for any $K \geq 1$. Therefore, the \composite inter-proposal time can be made arbitrarily small by adding component instances.
In real world deployments, component instances may occasionally fail to confirm a proposal for a slot, both in the good-case due to network delays and in the bad-case due to the presence of adversarial nodes. %
The \composite's confirmation rule turns these failures into head-of-line blocking events on the \composite log: a delayed or failed confirmation in one component protocol slot can prevent subsequent proposals already confirmed at the component level from being appended to the log.

This section identifies the \composite's optimal inter-proposal time $\epsilon^*$ and, in turn, the optimal number of parallel instances $K^*$, that minimize the expected transaction latency in the presence of proposals that fail to confirm.

\myparagraph{Stochastic Latency Model}
Let us fix a \composite protocol $\PIcomposite$ with component instances $\PIcomponent_1,...,\PIcomponent_K$ and component inter-proposal time $\deltaibt$.
Let $\epsilon = \deltaibt/K$ denote the \composite inter-proposal time, and index \composite proposal positions by $m=0,1,2,...$, with proposal time $T_m := T_0 + m\epsilon$ for some reference time $T_0 \geq \GST$.
For each position $m$, let $D_m \geq 0$ denote the random variable representing the time from $T_m$ until the component protocol proposing at position $m$ decides for the slot, either by confirming a payload or by assigning the empty value.

We assume $\{D_m\}_{m \geq 0}$ are independent and identically distributed (hereafter, i.i.d.) random variables with cumulative distribution function (CDF) $F$.
The \composite confirmation rule requires that the proposal at position $m$ is appended to the log only when all slots corresponding to positions up to and including $m$ have been decided by the component protocols. Consequently, we must account for the probability that all slots up to position $m$ are decided simultaneously. Because of the i.i.d. assumption, this joint probability decomposes into a product of per-slot probabilities, which reduces to a product of $F$ evaluations, where $F = \Prob{D_m \leq t}$.
More precisely, let us fix an arbitrary position $m$ and measure the time $t \geq 0$ from its proposal $T_m$ until the slot is decided.
Position $m - j$, with $j \geq 0$, was proposed at time $T_m - j\epsilon$ and, by definition of $D_{m-j}$, its component instance decides for it at time $(T_m - j\epsilon) + D_{m-j}$.
The event ``position $m - j$ has been decided by time $T_m + t$'' is therefore $(T_m - j\epsilon) + D_{m-j} \leq T_m + t$, which simplifies to $\{D_{m-j} \leq t + j\epsilon\}$.
By the i.i.d.\ assumption on $\{D_m\}$, these events are independent across $j$, so the steady-state probability that position $m$ is appended to the \composite log within time $t$ of its proposal is
\begin{equation}
\label{eq:H}
    H_\epsilon(t) \;:=\; \prod_{j=0}^{\infty} F(t + j\epsilon).
\end{equation}
For any $F$ with finite mean, the tail terms $1 - F(t + j\epsilon)$ are summable in $j$, so \cref{eq:H} converges.
Since $H_\epsilon$ is the CDF of a non-negative random variable, the survival-integral formula implies that the expected time from a position's proposal until it is appended to the \composite log is $\int_0^\infty \bigl(1 - H_\epsilon(t)\bigr)\, dt$.
We add the average time that a transaction waits to be included in the next available proposal, obtaining the steady-state expected transaction finality latency:
\begin{equation}
\label{eq:L}
    L(\epsilon)
    \;=\;
    \epsilon/2
    \;+\;
    \int_0^\infty\!\Bigl(1 - \prod_{j=0}^{\infty} F(t + j\epsilon)\Bigr) dt.
\end{equation}
Increasing $\epsilon$ widens the inclusion wait but reduces the integral, because every factor $F(t + j\epsilon)$ for $j \geq 1$ is evaluated at a larger argument and is therefore closer to one.
Intuitively, larger $\epsilon$ gives each predecessor more head start to decide, weakening head-of-line blocking effects.

\myparagraph{\composite's Optimal Inter-Proposal Time}
Each instance of the component protocol attempts deciding for a slot in time $\deltaconf$, where each attempt succeeds independently with probability $s \in (0,1)$ and fails with probability $p = 1 - s$. When nodes fail to decide for a slot, they proceed to the next slot and initiate a new decision attempt.
Let $R \in \{1, 2, \dots\}$ denote the number of attempts until the first success, so that $\Prob{R = r} = s p^{r-1}$ is the geometric probability mass function, and let $D_m \;:=\; \deltaconf \cdot R , \text{ with } D_m\in\; \{\deltaconf, 2\deltaconf, 3\deltaconf, \dots\}$.
Because $\{R > n\}$ is the event that all of the first $n$ independent attempts fail, which has probability $p^n$, the CDF of $D_m$ is the staircase $\Prob{D_m \leq n\deltaconf} \;=\; 1 - p^n$, with $n = \{1, 2, \dots\}$.
The staircase structure of the CDF introduces discontinuities at each $n\deltaconf$, rendering the product in \cref{eq:H} analytically intractable and precluding a closed-form optimization of \cref{eq:L} with respect to $\epsilon$. We bypass this by approximating $D_m$ by a random variable with a smooth $\Tconf$-shifted exponential CDF:
\begin{equation}
    \label{eq:Fp}
        F_p(t) \;:=\;
        \begin{cases}
            0, & t < \deltaconf, \\
            1 - p\, e^{-\lambda_p (t - \deltaconf)}, & t \geq \deltaconf,
        \end{cases}
        \qquad \lambda_p \;:=\; \frac{-\log p}{\deltaconf}.
    \end{equation}
The rate $\lambda_p$ is determined by requiring that $F_p$ exactly matches the staircase at every attempt boundary: substituting $t = n\deltaconf$ into \cref{eq:Fp} gives $F_p(n\deltaconf) = 1 - p e^{-\lambda_p(n-1)\deltaconf}$, and equating to $1 - p^n$ yields $\lambda_p = -\log p/\deltaconf$.
The shifted-exponential form rewrites the product in \cref{eq:H} as a standard \(q\)-series product. Specifically, setting $q := e^{-\lambda_p \epsilon} \in (0,1)$, we obtain for every \(u \geq 0\)
\begin{equation}
\label{eq:H-qpoch}
    H_\epsilon(\deltaconf + u)
    \;=\;
    \prod_{j=0}^{\infty}\!\bigl(1 - p e^{-\lambda_p u} q^j\bigr)
    \;=\;
    \bigl(p e^{-\lambda_p u};\, q\bigr)_\infty,
\end{equation}
where
\[
    (a;q)_\infty
    \;:=\;
    \prod_{m=0}^{\infty}(1-aq^m)
\]
denotes the \(q\)-Pochhammer symbol, the standard compact notation for infinite products with geometric spacing.
The next theorem gives the unique inter-proposal time $\epsilon^*$ that minimizes $L(\epsilon)$ under this model, in closed form.

\begin{theorem}[Optimal \composite Inter-Proposal Time]
\label{thm:gatling-exponential-optimum}
With $F_p$ as per \cref{eq:Fp}, the expected end-to-end transaction latency $L(\epsilon)$ in \cref{eq:L} has a unique minimizer
\begin{equation}
    \label{eq:epsilon-minimizer}
    \epsilon^*(s) = -\lambda_p^{-1}\log q_p^* = \deltaconf \cdot \frac{-\log q^*_p}{-\log p}
\end{equation}
where $q^*_p \in (0, 1)$ is the unique solution to
\begin{equation}
    \label{eq:q-implicit}
    \int_0^1 (p x;\, q^*_p)_\infty
    \sum_{j=1}^{\infty}
    \frac{p\, j\, (q^*_p)^j}{1 - p x (q^*_p)^j}\, dx
    \;=\;
    \frac{1}{2}.
\end{equation}
The corresponding minimum expected latency is
\begin{equation}
    \label{eq:minimized-latency}
    L^*(s) \;=\; \deltaconf + \lambda_p^{-1} \phi_p(q^*_p), \qquad \text{with}
\end{equation}
\begin{equation*}
    \phi_p(q) \;:=\; -\tfrac{1}{2}\log q + \int_0^1 \frac{1 - (p x;\, q)_\infty}{x}\, dx
\end{equation*}
\end{theorem}

\deferredproof{proof}{thm:gatling-exponential-optimum}

\myparagraph{Numerical Results}
We obtain $q^*_p$ by numerically solving \cref{eq:q-implicit} and then evaluate the closed-form expressions for $\epsilon^*(s)$ and $L^*(s)$ in \cref{thm:gatling-exponential-optimum}. \Cref{fig:gatling-final-plot} showcases $\epsilon^*(s)$ and $L^*(s)$ in units of $\Delta$ as functions of the per-attempt success probability $s$.
As $s \to 1$, the optimal inter-proposal time $\epsilon^*(s)$ shrinks to zero while the minimized expected latency $L^*(s)$ converges to $\deltaconf$: \composite recovers the component protocol's good-case confirmation latency with arbitrarily small inter-proposal time.
For example, at $s = 0.99$ the optimal global inter-proposal time is $\epsilon^*\approx 0.09\Delta$ and the corresponding expected end-to-end latency is $L^*\approx 3.09\Delta$. At $s = 0.999$ these become $\epsilon^*\approx 0.02\Delta$ and $L^*\approx 3.02\Delta$.
Translating $\epsilon^*$ into the design parameter $K$ of \cref{sec:protocol}, a \composite protocol that runs parallel instances of a component protocol with inter-proposal time $\deltaibt$ should choose a number $K^*(s) \;=\; \left\lceil \deltaibt / \epsilon^*(s) \right\rceil$ of parallel instances. Larger values of $K$ increase head-of-line blocking without further reducing latency.

\begin{figure}[t]
\centering
\begin{tikzpicture}
\begin{axis}[
    width=0.8\linewidth,
    height=0.5\linewidth,
    xlabel={Success probability $s$},
    ylabel={Latency (units of $\Delta$)},
    xmode=log,
    xmin=1e-3, xmax=1e-1,
    x dir=reverse,
    ymin=0,
    axis x line*=bottom,
    axis y line*=left,
    grid=both,
    xtick={1e-3,1e-2,1e-1},
    xticklabels={$0.999$,$0.99$,$0.9$},
    legend style={
        at={(0.5,-0.18)},
        anchor=north,
        legend columns=3,
        /tikz/every even column/.append style={column sep=0.5cm}
    },
]
\addplot[thick, mark=*] coordinates {
    (1e-1, 0.540798)
    (5e-2, 0.301679)
    (2e-2, 0.149065)
    (1e-2, 0.090358)
    (5e-3, 0.055871)
    (2e-3, 0.030280)
    (1e-3, 0.019311)
};
\addlegendentry{$\epsilon^*(s)$}
\addplot[dashed, thick, mark=square*] coordinates {
    (1e-1, 3.632369)
    (5e-2, 3.335710)
    (2e-2, 3.159178)
    (1e-2, 3.094587)
    (5e-3, 3.057689)
    (2e-3, 3.030894)
    (1e-3, 3.019585)
};
\addlegendentry{$L^*(s)$}
\addplot[dotted, thick] coordinates {
    (1e-3, 3)
    (1e-1, 3)
};
\addlegendentry{$3\Delta$ floor}
\end{axis}
\end{tikzpicture}
\caption{Optimal \composite global inter-proposal spacing $\epsilon^*(s)$ and minimized expected end-to-end finality latency $L^*(s)$ under a component protocol model with smooth approximation $F_p$ (\cref{eq:Fp}).}
\label{fig:gatling-final-plot}
\end{figure}
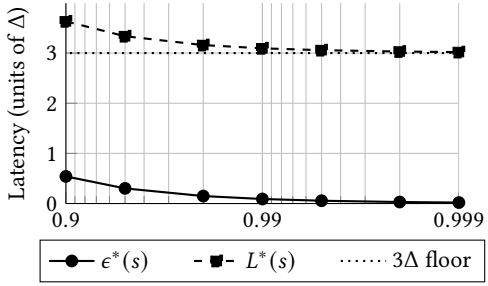

\begin{defer}{proof}
\begin{proof}[Proof of \cref{thm:gatling-exponential-optimum}]
We reduce the expected latency $L(\epsilon)$ to a single-variable function of $q := e^{-\lambda_p \epsilon} \in (0,1)$, then characterize the unique minimizer.

\emph{Step 1: latency in $q$-form.}
Because $F_p(t) = 0$ for $t < \deltaconf$, \cref{eq:H} gives $H_\epsilon(t) = 0$ for $t < \deltaconf$ as well, contributing $\deltaconf$ to $\int_0^\infty (1 - H_\epsilon(t))\, dt$ from the integration over $[0, \deltaconf)$.
Writing $t = \deltaconf + u$ with $u \geq 0$ and substituting \cref{eq:H-qpoch},
\[
    \int_0^\infty (1 - H_\epsilon(t))\, dt
    \;=\;
    \deltaconf \;+\; \int_0^\infty\!\Bigl(1 - (p e^{-\lambda_p u};\, q)_\infty\Bigr) du .
\]
Applying the change of variables $x := e^{-\lambda_p u} \in (0, 1]$, so that $du = -dx/(\lambda_p x)$, $u = 0$ corresponds to $x = 1$, and $u \to \infty$ corresponds to $x \to 0$,
\[
    \int_0^\infty (1 - H_\epsilon(t))\, dt
    \;=\;
    \deltaconf \;+\; \lambda_p^{-1} \int_0^1 \frac{1 - (p x;\, q)_\infty}{x}\, dx .
\]
Adding the inclusion wait $\epsilon/2 = -(\log q)/(2 \lambda_p)$,
\begin{equation}
\label{eq:L-phi}
    L(\epsilon)
    \;=\;
    \deltaconf \;+\; \lambda_p^{-1}\, \phi_p(q),
    \qquad
    \phi_p(q) \;:=\; -\tfrac{1}{2}\log q + \int_0^1 \frac{1 - (p x;\, q)_\infty}{x}\, dx .
\end{equation}
Since $\epsilon \mapsto q = e^{-\lambda_p \epsilon}$ is a strictly decreasing bijection from $(0, \infty)$ onto $(0, 1)$, minimizing $L(\epsilon)$ over $\epsilon \in (0, \infty)$ is equivalent to minimizing $\phi_p(q)$ over $q \in (0, 1)$.

\emph{Step 2: existence and uniqueness of the minimizer.}
As $q \to 0$ ($\epsilon \to \infty$), the inclusion-wait term $-\tfrac{1}{2}\log q$ diverges, hence $\phi_p(q) \to \infty$.
As $q \to 1$ ($\epsilon \to 0$), the $q$-Pochhammer $(p x;\, q)_\infty = \prod_{m \geq 0}(1 - p x q^m)$ tends to zero for every $x \in (0, 1]$, since each factor approaches $1 - p x < 1$, so the integrand of $\phi_p$ behaves like $1/x$ near $x = 0$ and the integral diverges.
Therefore $\phi_p$ attains an interior minimum on $(0, 1)$.
Strict convexity of $\epsilon \mapsto \Mean{\sup_{j \geq 0}(D_j - j\epsilon)}$ --- the maximum of linear functions of $\epsilon$, hence convex pathwise, with expectation preserving strict convexity --- yields uniqueness of the minimizer $q^*_p$, and consequently of $\epsilon^*(s) = -\lambda_p^{-1} \log q^*_p$.

\emph{Step 3: first-order condition.}
Differentiating the $q$-Pochhammer in $q$ under the integral sign,
\[
    q \, \partial_q (p x;\, q)_\infty
    \;=\;
    -(p x;\, q)_\infty \sum_{j=1}^{\infty} \frac{p x\, j\, q^j}{1 - p x q^j} .
\]
Substituting into $\phi_p'(q)$,
\[
    q\, \phi_p'(q)
    \;=\;
    -\tfrac{1}{2}
    \;+\;
    \int_0^1 (p x;\, q)_\infty \sum_{j=1}^{\infty} \frac{p\, j\, q^j}{1 - p x q^j}\, dx ,
\]
where the factor $x$ in the numerator of the previous display cancels against the $1/x$ in the integrand of $\phi_p$.
Setting $q\, \phi_p'(q) = 0$ at the interior minimizer $q^*_p$ yields the implicit equation in the statement.
Finally, substituting $\lambda_p = -(\log p)/\deltaconf$ into $\epsilon^*(s) = -\lambda_p^{-1} \log q^*_p$ gives the displayed closed form for $\epsilon^*(s)$, and evaluating \cref{eq:L-phi} at $q = q^*_p$ gives $L^*(s)$.
\end{proof}
\end{defer}

\section{Transaction Deduplication}
\label{sec:tx-deduplication}

Since \composite forwards every transaction to all $K$ component protocol instances (\alglocref{alg:composite-prot-optimized}{line:for-loop-k-instances-2,lst:line:broadcast-component}), the same transaction may appear multiple times in the \composite output log. As observed for multi-proposer protocols~\cite{sedna,mcp,bullshark}, duplicate transactions waste block space and reduce the effective throughput of the protocol. %
One might ask whether forwarding each transaction only to the component instance next in line to propose, rather than to all $K$ instances, would avoid duplicates and thereby improve what is sometimes called \emph{goodput}.
For practical purposes, this is indeed the case.
Modern blockchains frequently rely on private mempools or relay networks, which disseminate transactions rapidly and preferentially to upcoming leaders, effectively ensuring that transactions are promptly available to the next leader to propose.
In turn, for most consensus protocols, it suffices for liveness if an honest leader includes a transaction in its proposed block.

From a theoretical perspective,
\cref{def:ab-security} guarantees liveness of the component protocols only if a transaction is eventually input by all honest nodes.
This requires that for Gatling liveness, honest nodes forward input transactions to all component protocols.
On the other hand, strengthening the liveness guarantee so that input by one rather than all honest nodes suffices for liveness,
would introduce complications in the definition of latency: under the current \cref{def:ab-security}, the latency clock starts when all honest nodes have broadcast the transaction, ensuring that the next leader has already received it. If the clock instead starts when only one honest node has broadcast the transaction, the transaction may still need to propagate to the next leader or wait for the node that broadcast it to become a leader, potentially adding up to the effective transaction latency.
The theoretical latency metric then contains terms that in reality are not there due to how transaction dissemination and consensus are handled in modern systems.

\section{Enhancing \composite with Predictable Validity}
\label{sec:predictable-validity}

\begin{figure*}[t]
    \centering
            \includegraphics[width=\linewidth]{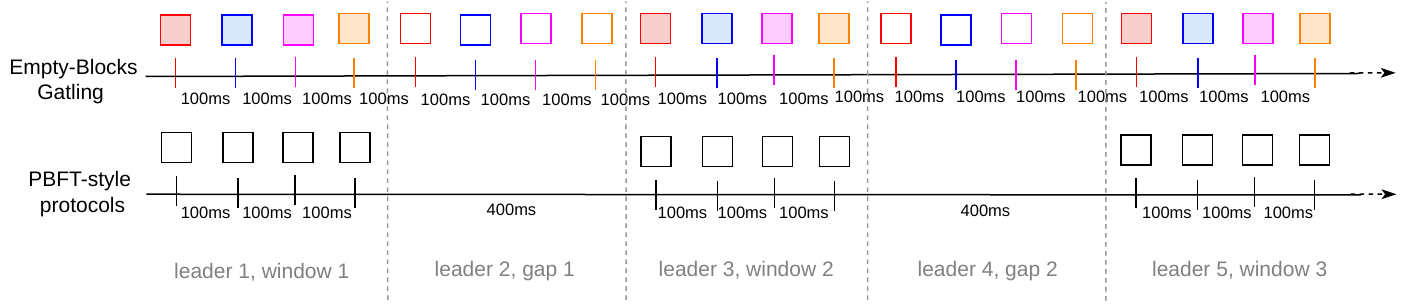}
        \caption{%
        Comparison between the empty-blocks \composite (top) and a PBFT-style protocol (bottom), both instantiated with a slowly rotating leader schedule. In this picture, $K=4$, $\Delta = 400$ms, and the \composite  inter-proposal time is $100$ms. The PBFT-style protocol proposes multiple (non-empty) blocks during a leader's window and separates consecutive windows by at least $\Delta$ to ensure block propagation and, therefore, predictable validity. The empty-blocks variant of \composite reproduces this behavior by alternating proposal windows with explicit $\Delta$-sized gaps with empty proposals.}
        \label{fig:simple-variant}
\end{figure*}

An attentive reader may have observed that when \composite is run with an inter-proposal time smaller than $\Delta$ and with successive blocks being proposed by different leader nodes (rotating leaders), a block does not have enough time to propagate from one leader to the next before the following proposal is issued. As a result, a leader cannot validate the transactions in its proposal against the most recent state of the log. Under these conditions, \composite no longer satisfies a property commonly referred to as predictable validity~\cite{10.1145/3558535.3559777,10.1007/978-3-031-18283-9_27,Kiffer2023NakamotoCU}.
Predictable validity guarantees that transactions deemed valid by a leader at proposal time %
remain semantically valid when they are eventually executed. Under predictable validity, proposers can therefore reason soundly about execution outcomes using only the locally available state. %
While most popular blockchain protocols such as Bitcoin, Ethereum, and Solana fulfill predictable validity, DAG-based logs such as Sui~\cite{sui-lutris} and Aptos~\cite{shoal} as well as ``lazy'' logs such as Celestia~\cite{celestia} and Pod~\cite{podnetwork} have forgone it in favor of lower latency. For blockchains that want to preserve predictable validity while retaining the benefits of reduced inter-proposal time, we present two variants of \composite---\emph{empty-blocks} and \emph{subprime-blocks}---with inter-proposal times smaller than $\Delta$, operating under slowly rotating leaders.

\subsection{Empty-Blocks \composite Variant}
\label{subsec:first-variant}

Consider a \composite protocol that achieves a sub-$\Delta$ inter-proposal time $\frac{\deltaibt}{K}$ by running $K$ parallel instances of a component protocol, each configured with an inter-proposal time of $\deltaibt$. %
Time is partitioned into a sequence of contiguous slots, called windows, of length at least $\Delta$. Each window is assigned to a randomly selected leader node which, within its assigned window, can issue a fixed number of consecutive proposals.
For \composite to achieve this slowly rotating leader schedule, the leader schedules of the component protocols are coordinated so that a single leader can issue $K$ consecutive \composite blocks by proposing in $K$ consecutive component protocol instances. This is illustrated for $K=4$ at the top of \cref{fig:simple-variant} by having the \composite log comprising of one block per color, with each color identifying one of the \composite's component protocols. To enable this coordination, each atomic broadcast component protocol exposes an interface that accepts the proposer schedule as an explicit input, in addition to the sequence of proposal times. %

Within a window, a sub-$\Delta$ inter-proposal time does not compromise predictable validity: because the same leader issues multiple consecutive proposals, it can locally maintain and update a consistent view of the log across proposals. Hence, when building a new proposal, the leader can validate the transactions against the most recent state of the log.
The challenge arises at window boundaries. Since \composite is instantiated with a sub-$\Delta$ inter-proposal time $\deltaibt$, less than $\Delta$ time elapses (at least) between the last proposal of one window and the first proposal of the subsequent window. Consequently, the leader of the subsequent window has to issue a proposal(s) before having received one or more proposals of the previous window.

To preserve predictable validity across windows, this variant of \composite introduces a ``window-gap-window-gap'' structure, with windows occurring in strict alternation with gaps of at least $\Delta$ (\cref{fig:simple-variant}). In a window, the assigned leader includes a non-empty payload in all of its proposals. In contrast, in a gap, proposals carry empty payloads that do not modify the state of the log. These gaps provide sufficient time for proposal propagation between consecutive windows, allowing the leader responsible for the next window to observe the preceding blocks. This ensures predictable validity.
The empty proposals could be issued similarly to the ones in a window, or they could be issued by a leader so that all proposals are received by the next leader before it begins issuing proposals for its window (e.g., unique slot proposal time $\Tproposal$ at the beginning of the window for all proposals). Alternatively, the protocol could mandate nodes to implicitly insert empty proposals in their local logs, without requiring explicit consensus on those (no QCs).

\Cref{fig:simple-variant} illustrates that the empty-blocks variant of \composite (top of the figure) closely mirrors the proposal pattern of traditional PBFT-style protocols with slowly rotating leaders (bottom of the figure). In PBFT-like protocols, a leader proposes multiple blocks within a window, and successive windows are separated by at least $\Delta$ time to allow blocks to propagate, thereby preserving predictable validity. Similarly, this \composite variant alternates between proposal windows and $\Delta$-sized gaps: proposals are issued at sub-$\Delta$ intervals within a window, while the intervening gaps ensure sufficient time for propagation.
\composite provides a general black-box composition framework for atomic broadcast that encompasses stable and slowly rotating leader protocols, such as PBFT and Alpenglow, as special cases. By adjusting its parameters, \composite can reproduce their execution patterns, while also enabling alternative operating regimes in which proposals are issued at sub-$\Delta$ inter-proposal times and consistently carry transactions, trading predictable validity for higher proposal rates.

\subsection{Subprime-Blocks \composite Variant}
\label{subsec:advanced-variant}

\begin{figure*}[t]
    \centering
            \includegraphics[width=\linewidth]{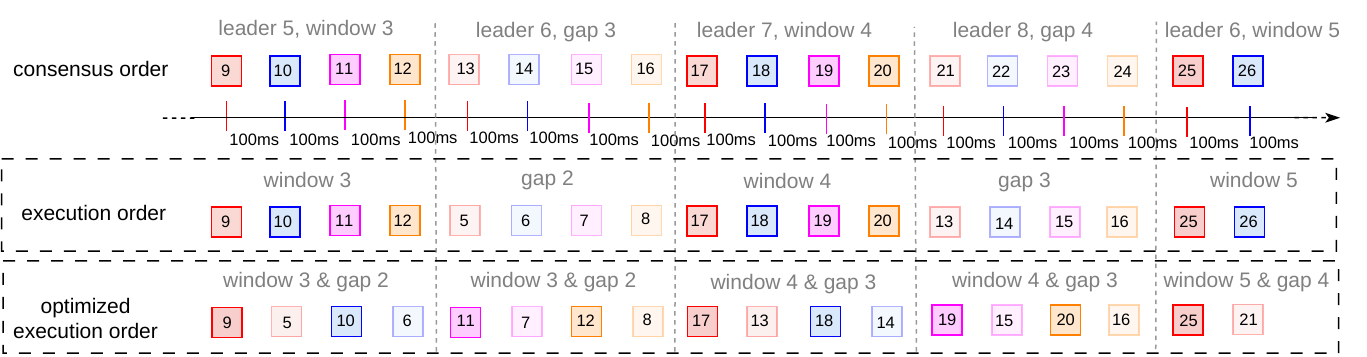}
        \caption{%
        Subprime-blocks variant of \composite that restores predictable validity at execution. At the top, the consensus confirmation ordering. At the bottom, we illustrate two possible orderings in which blocks are executed: %
        in one (top), blocks of windows and gaps are executed contiguously, with a window being executed before the preceding gap.
        In the other (bottom), for each window and the preceding gap, every block from the  window is executed before the corresponding block from the gap. }%
        \label{fig:advanced-variant}
\end{figure*}

While the empty-blocks variant of \composite already captures many of the desired properties (sub-$\Delta$ inter-proposal times and predictable validity), the need to validate transactions against the most recent state introduces gaps that periodically stall progress and increase the inter-proposal time of non-empty blocks. %
To address this limitation, we introduce an advanced variant of \composite named \emph{subprime-blocks} \composite---illustrated in \cref{fig:advanced-variant}---that preserves predictable validity while eliminating the stalls. %

The key observation is that empty blocks proposed in the gaps of the previous design can be repurposed as a distinct tier of block space, which we refer to as \emph{subprime block space}. This tier is reserved for transactions that are not time-critical and are insensitive to ordering races, such as simple value transfers or layer-2 data blob transactions. In contrast, the \emph{prime block space} tier is reserved for high-priority transactions---such as DeFi interactions---that are sensitive to execution order and frequently participate in front-run races. In this variant, the designated leader of a window makes proposals for prime block space, by only including time-critical transactions that access highly contentious state. Instead, the designated leader of a window of a gap makes proposals for subprime block space, by only including non-time-critical transactions that access non-contentious state.

Additionally, we modify the interaction between the consensus and execution layers by allowing the log produced by consensus to differ from the log consumed by the execution layer. %
As per \cref{fig:advanced-variant}, consider window number 3, followed by gap number 3, followed by window number 4, followed by gap number 4, and so on. While the \composite log output by consensus is the ordered sequence of blocks proposed in subsequent windows and gaps, %
the log consumed by the execution layer always executes a window before the immediately preceding gap. For example, the execution layer processes blocks from window 3, then blocks from gap 2, then blocks from window 4, followed by blocks from gap 3, and so on (see the execution order in \cref{fig:advanced-variant}). Because each gap has duration at least $\Delta$, when the leader of a window issues its first proposal, all blocks that will be executed before it must have been proposed at least $\Delta$ time earlier and therefore already received and executed by that leader. Consequently, proposers in windows always construct their proposals on top of a known and up-to-date log state.
In contrast, the leaders proposing during gaps have no knowledge of the state at the time their subprime blocks will be executed, since their blocks will be executed after the blocks proposed in the next window. This is acceptable because subprime blocks are state-independent and, therefore, always valid.
We can further refine the execution sequence by ordering blocks individually rather than executing all blocks from windows and gaps contiguously, as illustrated in the optimized execution ordering of \cref{fig:advanced-variant}: %
execution alternates between blocks from a window and blocks from the preceding gap. For instance, the first block of window 3 is executed first, then the first block of gap 2; next, the second block of window 3, followed by the second block of gap 2, and so on until all blocks of window 3 and gap 2 have been executed. The same alternating pattern is applied to window 4 and gap 3, etc. %

In the variant of \cref{subsec:first-variant}, nodes execute blocks only in during windows, while during gaps they wait for new blocks. This subprime-blocks \composite variant, besides higher proposal rates, also offers a more uniformly distributed execution load.

\section{Evaluation}
\label{sec:evaluation}

We implemented a proof-of-concept of \composite to measure its reduced transaction latency and its practicality in a realistic environment. Our experiment confirms the theoretical findings in \cref{sec:tail-latency}. %

Our implementation is available at:
\url{https://github.com/scaffino/Gatling}.

\myparagraph{Implementation Overview}
Our implementation extends Commonware's $\mathsf{alto}$ blockchain~\cite{commonware-alto}, a minimal but complete  Rust implementation of a blockchain that relies on the Commonware $\mathsf{monorepo}$~\cite{commonware-monorepo}, a collection of production-grade primitives including consensus, networking, and cryptographic libraries.
We augment $\mathsf{alto}$ by enabling each validator to run $K \geq 1$ independent instances of the Simplex consensus. Each instance maintains its own log, follows an independent leader schedule, and communicates over dedicated network channels. Importantly, each instance has an inter-proposal time of $\deltaibt = 500$ms and is configured to be non-responsive. We set the leader deadline to 225ms (the time by which the validators must have received the block; otherwise, they will cast a skip vote) and the notarize deadline to 375ms (the time by which a validator must have received a notarize certificate to cast a finalization certificate).
We implement the \composite log reconstruction and confirmation rules described in \cref{sec:protocol}, allowing the outputs of these parallel instances to be deterministically merged into a single global log.

We do not model transaction dissemination via gossip. Instead, transactions arrive according to a Poisson process at a rate of approximately 100 transactions per second. This transaction rate is below the system's saturation point and therefore does not introduce queuing effects that could distort latency measurements. Transactions are delivered directly to the validator expected to propose the next block, effectively modeling instantaneous transaction submission and emulating the private-mempool behavior commonly observed in real-world blockchain deployments. %
We note that when the \composite inter-proposal time $\epsilon$ is sufficiently small relative to the global network delay, bypassing the designated next leader may reduce latency when that leader is geographically distant.

To emulate crash failures, we randomly suppress proposals with probability $p$. Varying $p$ allows us to evaluate the effect of proposal failures on head-of-line blocking and transaction latency.

\myparagraph{Measuring Latency}
We deployed Gatling on a global 10 validator cluster with 1 node in London, 3 in Tokyo, 3 in Singapore, 2 in Dallas, and 1 in Miami. The largest latencies in the cluster were around 110-112ms between validators in Miami and Singapore. %
We compute the transaction latency for varying values of $K$.
We denote by skip rate the percentage of slots that resulted in a skip vote, i.e., the leader missed the 225ms deadline, aggregated over all instances during the measurement window.
\Cref{tab:latency-no-drops} shows \composite's inter-proposal time, skip rate, and transaction latency as a function of $K$ under a 0\% proposal-drop rate. \Cref{tab:latency-drop-1,tab:latency-drop-5,tab:latency-drop-10} showcase the average transaction wait time ($\epsilon/2$), confirmation time, and average latency for the 1\%, 5\%, and 10\% drop rate regimes, respectively.
\Cref{fig:results} presents the transaction latency as a function of $K$, measured across the three proposal-drop rate regimes: 1\% (right), 5\% (center), and 10\% (left). Each bar is a stacked bar chart decomposing total latency into its two components: the wait time (blue), and the confirmation time (orange).

\begin{figure*}[t]
    \centering
    \includegraphics[width=\linewidth]{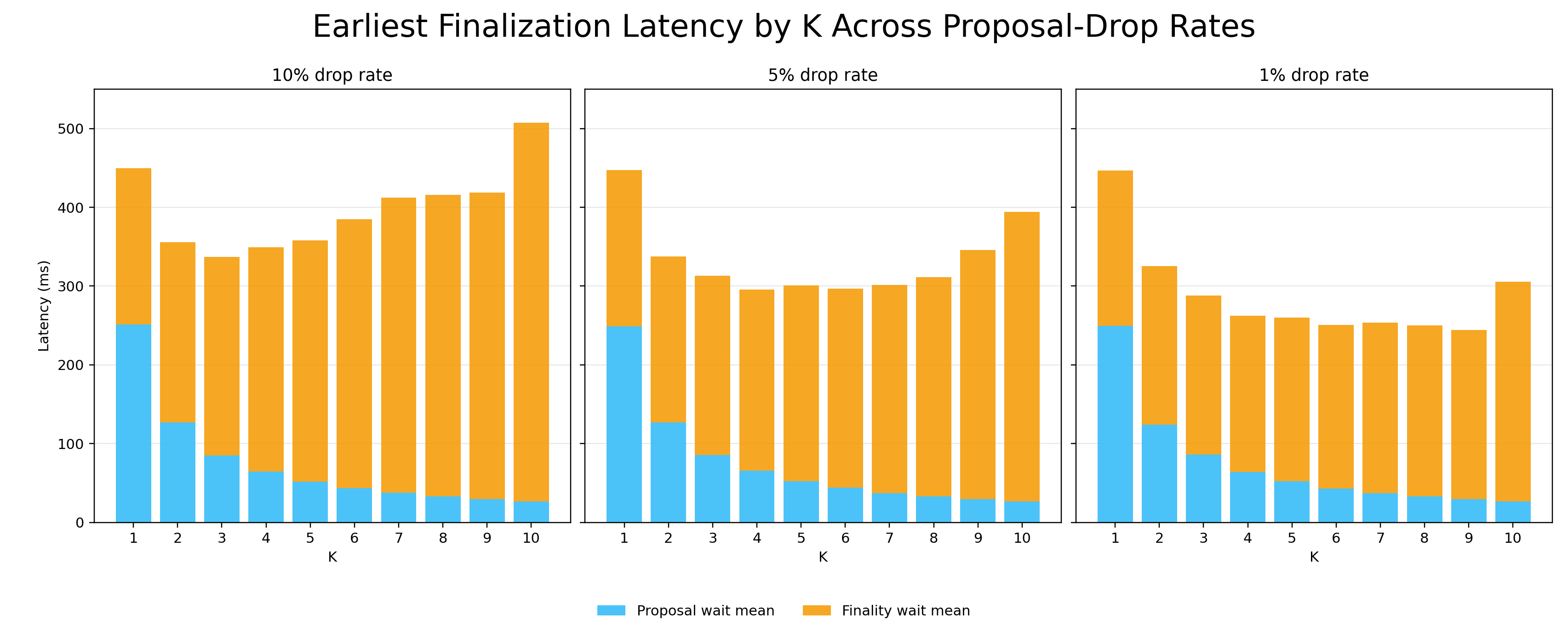}
    \caption{Transaction latency as a function of $K$ across three proposal drop rate regimes: 10\% (left), 5\% (center), and 1\% (right). Each bar decomposes total latency into the proposal wait time (blue), representing the \composite inter-proposal time $\epsilon$, and the finality wait time (orange), representing the confirmation time $\deltaconf$. The more reliable the validator set is, the more concurrent instances can be tolerated without introducing additional delay from head-of-line blocking.}
    \label{fig:results}
\end{figure*}

\begin{table}[t]
    \centering
    \small
    \setlength{\tabcolsep}{4pt}
    \caption{
    \composite's inter-proposal time, skip rate, and transaction latency as a function of the number of parallel consensus instances $K$
    under a 0\% proposal-drop rate.}
    \label{tab:latency-no-drops}
    \begin{threeparttable}
        \begin{tabular*}{\columnwidth}{@{\extracolsep{\fill}}cccc}
            \toprule
            \textbf{$K$} & \textbf{\makecell{Inter-Proposal \\ Time (ms)}} & \textbf{\makecell{Skip Rate \\ (\%)}} & \textbf{\makecell{Transaction \\ Latency (ms)}} \\
            \midrule
            1  & 502.2 & 0.3 & 450.4 \\
            2  & 251.3 & 1.2 & 326.8 \\
            3  & 167.6 & 0.3 & 282.2 \\
            4  & 125.7 & 0.3 & 258.9 \\
            5  & 100.5 & 0.3 & 249.9 \\
            6  & 83.8 & 0.4 & 239.5 \\
            7  & 71.8 & 0.3 & 233.4 \\
            8  & 62.9 & 0.3 & 231.2 \\
            9  & 55.9 & 0.3 & 226.6 \\
            10 & 50.3 & 0.3 & 223.3 \\
            12 & 41.9 & 0.4 & 222.1 \\
            15 & 33.7 & 4.7 & 366.7 \\
            20 & 25.1 & 0.3 & 217.0 \\
            25 & 20.1 & 0.3 & 215.8 \\
            40 & 12.6 & 3.5 & 411.4 \\
            50 & 10.1 & 0.3 & 214.7 \\
            \bottomrule
        \end{tabular*}
    \end{threeparttable}
\end{table}

\begin{table}[t]
    \centering
    \small
    \setlength{\tabcolsep}{4pt}
    \caption{
        Transaction wait time ($\epsilon/2$), confirmation time ($\deltaconf$), and latency as a function of $K$ with a 1\% proposal-drop rate.}
    \label{tab:latency-drop-1}
    \begin{threeparttable}
        \begin{tabular*}{\columnwidth}{@{\extracolsep{\fill}}cccc}
            \toprule
            \textbf{$K$} & $\epsilon /2$ (ms)& $\deltaconf$  (ms)  & \textbf{Transaction Latency} (ms) \\
            \midrule
            1  & 249.6 & 196.9 & 446.5 \\
            2  & 124.3 & 200.8 & 325.1 \\
            3  & 86.2  & 201.7 & 287.9 \\
            4  & 63.8  & 198.8 & 262.6 \\
            5  & 52.0  & 208.1 & 260.1 \\
            6  & 43.0  & 207.4 & 250.4 \\
            7  & 37.1  & 216.7 & 253.8 \\
            8  & 33.1  & 216.9 & 249.9 \\
            9  & 29.2  & 214.8 & 244.1 \\
            10 & 26.4  & 279.3 & 305.7 \\
            \bottomrule
        \end{tabular*}
    \end{threeparttable}
\end{table}

\begin{table}[t]
    \centering
    \small
    \setlength{\tabcolsep}{4pt}
    \caption{Transaction wait time ($\epsilon/2$), confirmation time ($\deltaconf$), and latency as a function of $K$ with a 5\% proposal-drop rate.}
    \label{tab:latency-drop-5}
    \begin{threeparttable}
        \begin{tabular*}{\columnwidth}{@{\extracolsep{\fill}}cccc}
            \toprule
            \textbf{$K$} & $\epsilon /2 $ (ms)& $\deltaconf$  (ms)  & \textbf{Transaction Latency} (ms) \\
            \midrule
            1  & 248.9 & 198.3 & 447.2 \\
            2  & 127.2 & 210.2 & 337.4 \\
            3  & 85.4  & 227.7 & 313.1 \\
            4  & 65.9  & 229.8 & 295.7 \\
            5  & 52.1  & 248.5 & 300.6 \\
            6  & 44.2  & 252.5 & 296.6 \\
            7  & 36.8  & 264.4 & 301.1 \\
            8  & 33.1  & 278.1 & 311.1 \\
            9  & 29.5  & 316.5 & 346.0 \\
            10 & 26.8  & 367.4 & 394.1 \\
            \bottomrule
        \end{tabular*}
    \end{threeparttable}
\end{table}

\begin{table}[t]
    \centering
    \small
    \setlength{\tabcolsep}{4pt}
    \caption{
    Transaction wait time ($\epsilon/2$), confirmation time ($\deltaconf$), and latency as a function of $K$ with 10\% proposal-drop rate.}
    \label{tab:latency-drop-10}
    \begin{threeparttable}
        \begin{tabular*}{\columnwidth}{@{\extracolsep{\fill}}cccc}
            \toprule
            \textbf{$K$} & $\epsilon /2$ (ms)& $\deltaconf$  (ms)  & \textbf{Transaction Latency} (ms) \\
            \midrule
            1  & 251.3 & 198.5 & 449.9 \\
            2  & 126.7 & 229.2 & 355.9 \\
            3  & 84.8  & 252.1 & 336.9 \\
            4  & 64.3  & 285.2 & 349.5 \\
            5  & 51.7  & 306.2 & 357.9 \\
            6  & 43.2  & 341.5 & 384.7 \\
            7  & 37.6  & 374.5 & 412.1 \\
            8  & 32.9  & 382.7 & 415.6 \\
            9  & 29.3  & 389.7 & 419.1 \\
            10 & 26.7  & 480.6 & 507.3 \\
            \bottomrule
        \end{tabular*}
    \end{threeparttable}
\end{table}

In the good case, where no proposal failures occur, increasing K reduces the inter-proposal time and thus lowers transaction waiting time, as shown in the second column of \cref{tab:latency-no-drops}. As a consequence, we expect the transaction latency to decrease monotonically with larger values of $K$. This is confirmed by \cref{tab:latency-no-drops}; for $K=\{15,40\}$ the average transaction latency shows a slight increase due to anomalously high skip rates. \Cref{tab:latency-no-drops} also reveals that end-to-end latency stabilizes around $215$ms, consistent with the theoretical $3\Delta$ confirmation time floor for $\Delta \approx 70$ms.

In the presence of proposal failures, transaction latency is expected to exhibit a non-monotonic dependence on K: it initially decreases as the inter-proposal time shrinks, but eventually increases as head-of-line blocking becomes the dominant source of delay (\cref{sec:tail-latency}). This is consistent with \cref{tab:latency-drop-1,tab:latency-drop-5,tab:latency-drop-10}, as well as with \cref{fig:results}: the transaction latency is minimized at $K=9$ for the 1\% drop rate, at $K=4$ for the 5\% drop rate, and at $K=3$ for the 10\% drop rate.
In the right panel of \cref{fig:results} the blue component shrinks steadily from $K=1$ to $K=9$, and total latency falls from $447$ms to a minimum of $244$ms at $K=9$. Beyond $K=9$, however, a 1\% proposal-drop probability is sufficient to trigger head-of-line blocking frequently enough that the confirmation time wait (orange component) grows, pushing total latency back up to $306$ms at $K=10$. The center and left panels of \cref{fig:results} illustrate how reliability governs the optimal $K$. With a 5\% drop rate (center), the optimal number of instances shifts to $K=4$, yielding a minimum latency of $296$ms, while with a 10\% drop rate the optimum is at $K=3$ with $337$ms. In both cases, the effect of head-of-line blocking increases with $K$. This confirms the core tradeoff identified in \cref{sec:tail-latency}. %

Our experiments also demonstrate that a blockchain can rotate leaders every $\epsilon = 50$ms or less on a global cluster where communication time between nodes exceeds 50ms.

\myparagraph{Consistency with the Theoretical Optimum}
The model of \cref{sec:tail-latency} makes three predictions about end-to-end latency as a function of the \composite inter-proposal time $\epsilon$; equivalently, as a function of the number of parallel instances $K = \deltaibt / \epsilon$.
First, the latency curve has a unique minimum at some $\epsilon^*(s)$.
Second, the corresponding optimum $K^*(s)$ shifts toward smaller $K$ as the per-attempt failure probability $p = 1 - s$ grows.
Third, as $s \to 1$, the minimum expected latency $L^*(s)$ approaches the single-attempt confirmation floor $\deltaconf$.
All three predictions are visible in \cref{fig:results}: each panel exhibits a clear U-shape with a minimum; the optimum shifts from $K = 9$ at the 1\% drop rate to $K = 4$ at the 5\% drop rate and $K = 3$ at the 10\% drop rate; and the minimum observed latency falls monotonically as the drop rate shrinks, from $337$ms at 10\% to $297$ms at 5\% and $244$ms at 1\%, approaching $\deltaconf \approx 197$ms, the time of one round of the component protocol, measured at $K = 1$ where head-of-line blocking is absent (\cref{tab:latency-drop-10,tab:latency-drop-5,tab:latency-drop-1}).

\section{Related Work}
\label{sec:relwork}

\myparagraph{Responsive Protocols and Timing Games}
An important line of work in distributed systems and blockchains has focused on a class of protocols, called \emph{responsive}, which improve latency by reducing dependence on worst-case timing assumptions~\cite{decentralized-thoughts-responsiveness,attiya-dwork-lynch-stockmeyer-1991}. Protocols such as Tendermint~\cite{buchman2016tendermint} and Streamlet~\cite{streamlet} advance according to fixed delays chosen to tolerate the maximum network delay ($\Delta$), causing nodes to wait for entire $\Delta$ timeout periods even when the actual message delay ($\delta$) is much smaller.

In contrast, responsive protocols such as  PBFT~\cite{pbft}, Simplex~\cite{simplex}, HotStuff~\cite{hotstuff}, HotStuff-1~\cite{hotstuff-1}, HotStuff-2~\cite{hotstuff-2}, Jolteon~\cite{jolteon}, and Moonshot~\cite{moonshot} are designed so that progress is driven by message arrivals rather than by timeout expirations, allowing nodes to reach a decision as quickly as the network conditions permit. %
As pointed out in~\cite{decentralized-thoughts-responsiveness,hybrid-consensus,herzberg-kutten-2000}, responsiveness is execution-dependent and cannot be guaranteed in all runs: when a leader is adversarial, correct nodes may be forced to delay commitment, reintroducing latency that is independent of the actual network conditions. As a result, responsiveness holds only for a subset of optimistic executions: those with a correct leader and sufficiently many honest participants.

For a long time, the design of responsive protocols has been carried out without considering timing incentives of the nodes. Only recently, the work of~\cite{time-is-money,chorusone-twitter}, analyzed the incentives of responsive protocols in the blockchain setting, showing that when elected as proposer the best strategy for profit-maximising nodes is to delay the proposal for as long as possible, while still ensuring the timely inclusion in the ledger.
This strategic delay in block proposals, often called timing game, is exploited to acquire additional information about pending transactions and optimize the Maximal Extractable Value (MEV) capture. Timing games have been observed in major blockchains such as Ethereum and Solana. Designing responsive protocols that explicitly account for timing games and strategic proposer behavior was first explored by~\cite{responsive-consensus}; however, this composition remains largely underexplored and represents a promising direction for future work.
Our work improves latency via a mechanism that is orthogonal and complementary to the one of responsive protocols. While responsive protocols achieve low latency under favorable network delays ($\delta$), \composite attains arbitrarily small inter-block times using only conservative worst-case assumptions ($\Delta$) in its component protocols, providing latency improvements independent of optimistic network conditions.

\myparagraph{Pipelined and Optimistic Protocols}
In traditional BFT-style protocols, block proposals are strictly sequential: a leader must wait for the previous proposal to be confirmed through multiple voting phases before issuing the next proposal. PBFT~\cite{pbft}, for example, proceeds through a leader proposal phase (pre-prepare) followed by two voting phases (prepare, commit), resulting in a ``propose-vote-vote'' structure per proposal. With Casper-FFG~\cite{casper-ffg} and then the HotStuff family of protocols~\cite{hotstuff,hotstuff-1,hotstuff-2}, a \emph{pipelined}, \emph{chain-based} formulation is introduced: proposals are cryptographically linked to each other in a chain, and each proposal carries a quorum certificate attesting that a quorum of nodes has agreed on the previous proposal. This structure allows consensus phases to be pipelined: while nodes vote to form a certificate for the current proposal, the leader can already issue a new proposal. Conceptually, this results in an alternating ``propose-vote-propose-vote'' pattern, where making a new proposal is overlapped with voting on the previous one.

A complementary line of work explores optimistic/speculative techniques that further reduce the number of network round trips to reach confirmation in BFT consensus protocols. For instance, Moonshot~\cite{moonshot} and Hydrangea++~\cite{hydrangeaplusplus} adopt an \emph{optimistic proposal} approach allowing leaders to propose new blocks before the previous block has been quorum-certified, overlapping proposal and voting phases. This reduces latency under the assumption of consecutive honest leaders. This design effectively results in a ``propose-propose-propose'' structure, enabling more aggressive proposal pipelining. We observe that for $K=2$, \composite recovers the behavior of Moonshot in the optimistic case: when the leader is honest and includes a transaction in its block, and the previous leader was also honest, then that transaction is viewed as confirmed by everyone $3\Delta$ later, even if at most $f$ nodes are adversary.

A different approach is taken by HotStuff-1~\cite{hotstuff-1}: next to the traditional HotStuff-style commit rule, in the optimistic case, it introduces a speculative technique that reduces commit latency by sending clients early execution responses after a single quorum certificate, without waiting for the usual multi-phase commit rule.

\myparagraph{Fast-Path Protocols}
A long line of work that started with~\cite{fab} has explored \emph{fast-path} techniques that reduce latency by decreasing the number of voting rounds required to confirm a block, trading resilience for lower round complexity. Classical BFT protocols require two rounds of voting to tolerate up to $f$ Byzantine faults, ensuring safety even under adversarial conditions. Fast-path designs relax these guarantees under optimistic assumptions---such as honest leader and a reduced number of adversarial nodes---to decide on a proposal with a single round of voting. This approach reduces transaction latency when the assumptions hold, but retain a fallback slow-path with two rounds of voting that restores full fault-tolerance when the optimistic assumptions are violated.
Recently, there has been renewed interest in fast-path designs in blockchains and distributed systems which has led to the designs of Minimmit~\cite{minimmit}, Alpenglow~\cite{alpenglow}, Kudzu~\cite{kudzu}, and Hydrangea~\cite{hydrangea}.
While fast-path protocols reduce latency by weakening fault-tolerance assumptions in optimistic executions, \composite preserves full resilience while improving block production rates through composition. As such, fast-path techniques are complementary to \composite rather than competing: they could be layered on top of the component protocols to further reduce \composite's confirmation latency. %

\myparagraph{Multiple Concurrent Proposals}
Protocols that leverage multiple concurrent proposals (MCP) have recently gained renewed attention in the Ethereum and Solana communities, notably through the Ethereum's EIP-7805 (FOCIL)~\cite{focil} and this recently proposed protocol~\cite{mcp}. MCP approaches allow multiple nodes to simultaneously propose batches of transactions for inclusion in the next block; while each round still produces a single block, that block aggregates contributions from several proposers, thereby reducing the leader's monopoly over transaction inclusion and ordering. While \composite reduces the inter-proposal time, MCP protocols~\cite{mcp} require multiple rounds of communication among nodes when a proposal is constructed, which leads to longer inter-proposal times.

In DAG-based protocols~\cite{bbca,dag-rider,bullshark,tusk,shoal,mysticeti}, in each round, every node can independently create and broadcast a proposal that back-references to typically $2f+1$ proposals from the previous round, forming a directed acyclic graph (DAG). These protocols exploit parallelism by allowing multiple leaders to be active concurrently and they can be viewed as a form of parallel composition, in which multiple ordering streams with different leader schedules operate concurrently over a shared DAG. Unlike \composite, however, these streams remain coupled through the DAG's reference structure, which imposes dependencies and constrains progress to the network delay.

\myparagraph{Comparing \composite with Multi-BFT Protocols}
Recent work on Multi-BFT protocols explores running multiple PBFT instances in parallel to overcome the throughput limitations of single-leader designs.
For example, Mir-BFT~\cite{mirbft} increases throughput by running multiple PBFT proposers in parallel, each responsible for proposing a subset of log positions. To distribute work across proposers, transactions are partitioned into buckets and assigned to proposers according to a deterministic mapping. Confirmed proposals are merged into a single total order according to their assigned slot numbers. Mir further employs an epoch-based reconfiguration mechanism in which a designated epoch leader proposes how transactions, slots, and ordering responsibilities are distributed across proposers. This architecture removes the single-leader bottleneck of PBFT, but introduces an additional control layer for transaction partitioning, proposer coordination, and reconfiguration. \composite also exploits parallelism through multiple concurrent consensus instances, but does so without sharding transactions, assigning portions of the slot space to different proposers, or requiring a separate coordination layer. The resulting design is simpler and, through staggered proposal schedules, can arbitrarily reduce the inter-proposal time and, in turn, the transaction latency.

ISS~\cite{iss} generalizes Mir-BFT by introducing a protocol-independent framework for parallel consensus. Transactions are partitioned across multiple consensus instances, whose outputs are deterministically merged into a single replicated log. This abstraction removes the need for Mir's dedicated epoch leader and orderer-assignment mechanism while retaining the scalability benefits of parallel execution. Like ISS, \composite composes multiple consensus instances and combines their outputs into a single log. However, ISS and \composite target different goals. ISS uses parallelism to partition workload across independent streams and thereby increase throughput, whereas \composite uses staggered proposal schedules to reduce transaction latency.

To mitigate the impact of slow instances, Ladon~\cite{ladon} runs parallel BFT  protocols and constructs a single globally ordered log by assigning each block a monotonic rank. When a leader is about to propose a block, it first collects from $2f+1$ nodes the ranks of the most recently confirmed blocks across all instances, and sets its block's rank to be strictly greater than the collected ranks.
The global log is constructed by ordering blocks by increasing rank, breaking ties by instance number. Specifically, each node locally computes a threshold (the rank of the lowest-ranked last-confirmed observed block across instances, incremented by one) and confirms all blocks with rank below it.
Unfortunately, Ladon has important shortcomings. Specifically, Ladon never discusses how a block may have been confirmed in some node's local view, yet arrive at other nodes only after they have already confirmed blocks with higher ranks. Such an in-flight block would land below the threshold used to confirm those blocks, potentially contradicting the global order that nodes believed to be final. This constitutes a concrete safety violation that can arise both before and after $\GST$, and is not covered by Ladon's safety proof. Moreover, because leaders must collect cross-instance rank certificates and embed the rank in each block, Ladon's instances are neither independent nor closed-box.

\myparagraph{Comparing \composite with Shoal++}
After completing this work, we were made aware that Shoal++~\cite{shoalplusplus} uses a construction similar to \composite in the context of DAG-based protocols to reduce the inter-proposal time.
In particular, it runs three DAG instances in parallel, time-staggered by a fixed $\Delta$ offset, in a not-fully-generic ``ajar-box'' composition.
By contrast, \composite explores regimes with inter-proposal intervals below the network delay and introduces variants that preserve predictable validity even when the offsets become arbitrarily small. Moreover, \composite is closed-box and works for general atomic broadcast component protocols, and our experiments evaluate the latency reduction solely attributable to time-staggered parallel composition.
\begin{acks}
We thank
Ittai Abraham,
Adam Alon,
Brendan Chou,
Pranav Garimidi,
Aniket Kate,
Patrick O'Grady,
Ling Ren,
Dana Shamir,
Nibesh Shrestha,
and
Aviv Zohar
for fruitful discussions.
The work of Giulia Scaffino was conducted in part while at a16z Crypto Research.

\end{acks}

\bibliographystyle{ACM-Reference-Format}
\bibliography{references.bib}

\appendix %

\deferredsection{proof}{Deferred Proofs}

\end{document}